\documentclass[journal]{IEEEtran}
\IEEEoverridecommandlockouts
\def\BibTeX{{\rm B\kern-.05em{\sc i\kern-.025em b}\kern-.08em
		T\kern-.1667em\lower.7ex\hbox{E}\kern-.125emX}}
\usepackage{cite}
\usepackage{amsmath,amssymb,amsfonts}
\usepackage{algorithmic}
\usepackage{graphicx}
\usepackage{color}
\usepackage{textcomp}
\usepackage{amsthm}

\usepackage{chngcntr}
\usepackage{verbatim}
\usepackage{mathtools}
\usepackage{subfigure}
\usepackage{empheq}

\usepackage{enumitem}  
\graphicspath{{Graphics_pdf}}
\usepackage[normalem]{ulem}
\usepackage{cases}
\usepackage{steinmetz}
\usepackage{fancyhdr}
\usepackage[yyyymmdd,hhmmss]{datetime}
\usepackage[implicit=false]{hyperref}
\usepackage{glossaries}
\makeglossaries
\glsdisablehyper
\usepackage{comment}
\excludecomment{underconstruction}

\newcounter{eqJ}

\newtheorem{theorem}{Theorem}
\newtheorem{corollary}{Corollary}
\newtheorem{lemma}{Lemma}

\newtheorem{remark}{Remark}

\DeclareMathOperator{\rem}{\mathrm{mod}}

\newcommand{\rom}[1]{\uppercase\expandafter{\romannumeral #1\relax}}

\newcommand{\todo}[1]{\color{red}#1\color{black}}

\newcommand{\fig}{Fig.}    
\newcommand{\secR}{Section}

\newcommand{\gt}{g^{\mathrm{s}}}  
\newcommand{\gr}{g^{\mathrm{r}}}  
  
\newcommand{\omt}{\Omega^{\mathrm{s}}}  
\newcommand{\omr}{\Omega^{\mathrm{r}}}  
\newcommand{\phr}{\alpha^{\mathrm{r}}}  
\newcommand{\pht}{\alpha^{\mathrm{s}}}  
\newcommand{\phrS}{\tilde{\alpha}^{\mathrm{r}}}  
\newcommand{\phtS}{\tilde{\alpha}^{\mathrm{s}}}  
\newcommand{\phrVp}{\boldsymbol{\alpha}_p^{\mathrm{r}}}

\newcommand{\phrV}{\boldsymbol{\alpha}^{\mathrm{r}}}  
\newcommand{\phtV}{\boldsymbol{\alpha}^{\mathrm{s}}}  
\newcommand{\phrSV}{\tilde{\boldsymbol{\alpha}}^{\mathrm{r}}}  
\newcommand{\phtSV}{\tilde{\boldsymbol{\alpha}}^{\mathrm{s}}} 
\newcommand{\psrV}{\boldsymbol{\psi}^{\mathrm{r}}}  
\newcommand{\pstV}{\boldsymbol{\psi}^{\mathrm{s}}}  
\newcommand{\psrVp}{\boldsymbol{\psi}_p^{\mathrm{r}}}  
\newcommand{\phrI}{\beta^{\mathrm{r}}}  
\newcommand{\phtI}{\beta^{\mathrm{s}}}  
\newcommand{\Gr}{G_{\mathrm{r}}} 
\newcommand{\Gt}{G_{\mathrm{s}}} 

\newcommand{\psr}{\psi^{\mathrm{r}}}  
\newcommand{\pst}{\psi^{\mathrm{s}}}  
 
\newcommand{\Kr}{K_{\mathrm{r}}} 
\newcommand{\Lt}{L_{\mathrm{s}}}
\newcommand{\Lr}{L_{\mathrm{r}}}
\newcommand{\Lrp}{L_{\mathrm{r},p}}
\newcommand{\Ltm}{L_{\mathrm{s,max}}}
\newcommand{\Lrm}{L_{\mathrm{r,max}}}

\newcommand{\phir}{\phi^{\mathrm{r}}}
\newcommand{\phit}{\phi^{\mathrm{s}}}

\newcommand{\Td}{{T_\mathrm{d}}}
\newcommand{\Tsys}{T}

\newcommand{\snrA}{\gamma^\mathrm{(a)}}
\newcommand{\snrB}{\gamma^\mathrm{(b)}}
\newcommand{\SsnrA}{S_\mathrm{a}}
\newcommand{\SsnrB}{S_\mathrm{b}}

\newcommand{\SsnrD}{S_\mathrm{d}}

\newcommand{\JA}{J_\mathrm{a}}
\newcommand{\JB}{J_\mathrm{b}}

\newcommand{\SsnrAp}{S_{\mathrm{a},p}}
\newcommand{\JAp}{J_{\mathrm{a},p}}

\newcommand{\Xstb}{\mathcal{X}_{\Kr}^*}

\newcommand{\PsetR}{\mathcal{P}_{\text{r}}}
\newcommand{\PsetT}{\mathcal{P}_{\text{s}}}
\newcommand{\DsetX}{\mathcal{D}_{X}}
\newcommand{\DsetY}{\mathcal{D}_{Y}}
\newcommand{\Av}{\mathbb{\mathcal{A}}}

\newcommand{\phiVrs}{\boldsymbol{\phi}}

\newcommand{\xr}{x}  
\newcommand{\xrV}{\boldsymbol{x}} 
\newcommand{\xrVp}{\boldsymbol{x}^{(p)}}  

\newcommand{\xt}{z} 
\newcommand{\xtV}{\boldsymbol{z}}  

\newcommand{\yr}{y} 
\newcommand{\yrV}{\boldsymbol{y}}  
\newcommand{\yrVp}{\boldsymbol{y}^{(p)}} 

\newcommand{\vt}{v} 
\newcommand{\vtV}{\boldsymbol{v}}

\newcommand{\pPt}{u} 

\newacronym{AoI}{AoI}{age-of-information}
\newacronym{IRT}{IRT}{inter-reception time}
\newacronym{BrEP}{BrEP}{burst error probability}
\newacronym{SNR}{SNR}{signal-to-noise ratio}
\newacronym{SBR}{SBR}{single bounce reflection}
\newacronym{LOS}{LOS}{line of sight}
\newacronym{AOA}{AOA}{angle of arrival}
\newacronym{AOD}{AOD}{angle of departure}
\newacronym{ACN}{ACN}{analog combining network}
\newacronym{ASN}{ASN}{antenna switching network}
\newacronym{ABN}{ABN}{analog beamforming network}
\newacronym{VU}{VU}{vehicular user}
\newacronym{CAM}{CAM}{cooperative awareness message}
\newacronym{MRC}{MRC}{maximal ratio combining}
\newacronym{RF}{RF}{radio frequency}
\newacronym{CSI}{CSI}{channel state information}
\newacronym{C-ITS}{C-ITS}{cooperative intelligent transportation systems}
\newacronym{MIMO}{MIMO}{multiple-input multiple-output}
\newacronym{V2V}{V2V}{vehicle-to-vehicle}
\newacronym{Tx}{Tx}{transmitter}
\newacronym{Rx}{Rx}{receiver}
\newacronym{OFDM}{OFDM}{orthogonal frequency division multiplexing}
\newacronym{CDF}{CDF}{cumulative density function}

\begin{document}
\title{Robust Analog Beamforming for Periodic Broadcast V2V Communication\\
	\thanks{This research has been carried out in the antenna systems center \emph{ChaseOn} in a project financed by Swedish Governmental Agency of Innovation Systems (Vinnova), Chalmers, Bluetest, Ericsson, Keysight, RISE, Smarteq, and Volvo Cars.}%
\thanks{The authors are with the Communication Systems Group, Department of Electrical Engineering, Chalmers University of Technology, 412 96 Gothenburg, Sweden (e-mail: chouaib@chalmers.se; fredrik.brannstrom@chalmers.se; erik.strom@chalmers.se)}
}
\author{Chouaib~Bencheikh~Lehocine,  Fredrik~Br{\"{a}}nnstr{\"{o}}m, and Erik~G.~Str{\"{o}}m,~\IEEEmembership{Fellow,~IEEE}}
\maketitle

\begin{abstract}
We generalize an existing low-cost analog signal processing concept that takes advantage of the periodicity of vehicle-to-vehicle broadcast service to the transmitter side. In particular, we propose to process multiple antennas using either an \gls{ABN} of phase shifters, or an \gls{ASN} that periodically alternate between the available antennas, to transmit periodic messages to receivers that have an \gls{ACN} of phase shifters, which has been proposed in earlier work. 
To guarantee robustness, 
we aim to minimize the burst error probability for the worst receiving vehicular user, in a scenario of bad propagation condition that is modeled by a single dominant path between the communicating vehicles. 
In absence of any form of channel knowledge, we analytically derive the optimal parameters of both \gls{ABN} and \gls{ASN}. The \gls{ABN} beamforming vector is found to be optimal for all users and not only for the worst receiving user. 
 Further, we demonstrate that Alamouti scheme for the special case of two transmit antennas yields similar performance to \gls{ABN} and \gls{ASN}. At last, we show that the derived parameters of the two proposed transmission strategies are also optimal when hybrid \gls{ACN}-maximal ratio combining is used at the receiver.
 \end{abstract}
\glsresetall

 

\section{Introduction}
	\IEEEPARstart{V}{ehicular} communication 
	paves the path for the realization of \gls{C-ITS}.
	By allowing vehicles to share real-time information about their status, vehicles can cooperate and coordinate their movement and maneuvers, which results in increased safety, efficiency, comfort and sustainability of transportation systems. 
	Since \gls{C-ITS} requires the exchange of time-sensitive, critical information, very high reliability and low latency need to be supported by the vehicular communication systems. One typical technical solution to those requirements is the use of multiple antenna systems. In the context of vehicular communications, antennas pose their own challenges.
	It has been noted in several studies, including~\cite{AntPlac2007,AntPlac_2_2014,AntPlac2010}, that antenna patterns are distorted by several factors including vehicle body and mounting position. Such distortions can lead to very low gains or even blind spots at certain directions, which may result in low performance or outage when the transmitted or received signal is along those directions. 
	To enable a robust \gls{V2V} communication against the effects of such disturbances, multiple antennas can be processed with the objective of ensuring certain performance in the worst-case propagation scenario with respect to the antenna system. 
	
	The \gls{ACN} proposed in~\cite{ACN}, is such a solution that was designed to combine antennas at the receiver to achieve robustness. In particular, it was designed to minimize outage probability of a \gls{C-ITS} application, measured through the loss of $K$ consecutive periodic status messages, when the received signal coincides with the worst-case \gls{AOA}. We note that this metric is related to \gls{AoI} assessment metric of broadcast periodic communication~\cite{Aoinfo}, defined as the age of the information contained in the last correctly received periodic message.
	The \gls{ACN} was based on pure analog combining. To leverage on the digital processing benefits a hybrid analog-digital solution was presented in our previous work~\cite{HC}. 
	In both~\cite{ACN} and~\cite{HC}, the transmitter side has not been considered. For a comprehensive multiple antenna system we would like to explore, in accordance with a receiver that uses an \gls{ACN}, what beamforming solutions can be used at the transmitter side to improve robustness in a scenario of \gls{V2V} communication. 
	As we are considering a broadcast transmission, feeding back the \gls{CSI} may be infeasible due to the high number of \glspl{VU}. 
	Also, every \gls{VU} has only limited number of antennas and limited processing capabilities to beamform, based on \gls{CSI}, to all receiving \glspl{VU} at the same time.
	Therefore, herein, we target a transmit beamformer that does not depend on \gls{CSI}. This implies a low-complexity solution.
	Moreover, we target a low-cost analog beamforming solution.
	
	A selection of publications that are relevant to the scope of our work is~\cite{curvSA,OnOff,RBF}.
	In~\cite{curvSA}, 
	two transmission strategies were proposed and evaluated through measurement in a platoon scenario. The first strategy is based on alternating the transmit antennas periodically. The second uses information about the road curvature and selects the antenna with higher probability to have a \gls{LOS} with the receiving antennas. 
	The authors assess both strategies using an \gls{AoI} approach. The proposed schemes consider only the platoon vehicles. Other vehicles on the road, to which cooperative messages have to reach, are not considered. 
	In~\cite{OnOff}, another transmit beamforming structure based on switches is proposed. The scheme is fully analog and it is a variant of antenna selection. Instead of selecting a single antenna element, the transmitter chooses a subset of the antennas that results in maximizing the \gls{SNR} at the receiver. However, the scheme relies on \gls{CSI}, and as explained earlier, for our scenario of broadcast transmission such approach is not very relevant.
	Another work of interest is a random beamforming technique proposed in~\cite{RBF}. A uniform linear array of antennas is weighed by a vector that is randomized over time frequency blocks. The achieved average pattern over many time frequency blocks is omnidirectional. This scheme does not require channel knowledge, but it uses several \gls{RF} chains and therefore, it is a digital strategy.
	Besides this, standard digital approaches that do not depend on \gls{CSI}, e.g., Alamouti transmit diversity~\cite{Alamouti} and similar space-time or space-frequency codes are relevant in this scenario, however, they require the use of multiple \gls{RF} chains. We are interested in finding out how can we improve the system in analog domain, which is characterized by low cost. The digital solutions can be used on top to give enhanced performance and hybrid structures that achieve a trade-off between performance and cost. 
	
	In this paper, we assume that \glspl{VU} use an analog combining network as proposed in~\cite{ACN} at the receiver, while we propose two strategies at the transmitter side that do not rely on \gls{CSI}. The first one is an \gls{ABN} of phase shifters that has a similar construction to the receiver \gls{ACN}. The second strategy is an \gls{ASN} and it is based on alternating between the transmit antennas in a periodic manner. We note that this switching approach was used in~\cite{curvSA} as well, however, it did not take into account a receiver structure as the one proposed here. Furthermore, it was assessed using \gls{AoI} which is related to, but not the same as our assessment metric. 
	 Given the proposed schemes, we optimize the overall system parameters at both transmitter and receiver for the \gls{ABN}, while for the \gls{ASN} only the receiver parameters are optimized. Our optimization is based on the minimization of \gls{BrEP} of $K$ consecutive packets for the worst user in the system i.e., the user experiencing the worst \gls{BrEP}, which under some assumptions, can be defined by the worst-case \gls{AOA} and \gls{AOD}. That ensures a robust communication against unfavorable angles with respect to the antenna system. 
	 Once the optimal parameters are found, we come to show that the developed networks can be further improved by adding a digital processing stage at the receiver in the form of a hybrid combiner with similar structure to~\cite{HC}. The summary of the contributions of this paper follows.
\begin{itemize}
	\item We present two fully analog transmission strategies (\gls{ABN} and \gls{ASN}) that does not require any channel knowledge, in combination with an \gls{ACN} at the receiver.
	\item We provide the optimal transmitter and receiver parameters (phase slopes) associated with both strategies. These parameters minimize the \gls{BrEP} for the worst receiving \gls{VU}, under certain assumptions. 
	\item We extend the optimality proof of the phase slopes of the \gls{ACN} developed in~\cite{ACN}. The phase slopes were shown to be optimal only for systems with a number of receive antennas $\Lr \in\{2,3\}, \Lr\leq K$, here we demonstrate that they are optimal for any system with $\Lr\leq K$.
		\item We demonstrate that the analog structures can be upgraded by a digital processing stage at the receiver, based on \gls{MRC}, to yield enhanced performance.
\end{itemize}

	\section{System Model}
		In this section, we present the system model and the transmission strategies we are considering together with the receiver structure. Moreover, we state the main assumptions that we use to deduce the optimal parameters of the proposed strategies. 
	\subsection{Antennas and Channel Model}
	Scarce multipath propagation with a single dominant component---\gls{LOS} or \gls{SBR}---between the communicating vehicles, pose a challenging scenario for the robustness of a V2V communication based on some antenna system.
	In such scenario, the \gls{AOD} and \gls{AOA} may coincide with very low gain of the antenna patterns, and lead to packet loss, or outage if the same directions are approximately sustained for a time span of several consecutive packets.
Such propagation conditions have been noted to be prominent in traffic scenarios where the road is not surrounded by buildings~\cite{channel_v2v_angular}, e.g., highway. Moreover, it has been noted that in such propagation environments the azimuth angular spread is small~\cite{channel_v2v_angular}, which implies that the few multipath components are restricted to a narrow sector of the antenna system. Therefore, a good framework to ensure a robust \gls{V2V} communication against unfavorable \gls{AOD}, \gls{AOA} is to use a channel model that assumes such propagation conditions.
In the model to follow, however, we consider only the dominant path between the communicating vehicles, since it contributes to most of the received power.
 
 Let $\gt_m $ 
 and $\gr_l$ be the far field functions of the $m^{\text{th}}$ transmit and $l^{\text{th}}$ receive antennas\footnote{Throughout the paper, the superscript letter 's' stands for sender, while 'r' stands for receiver.}, respectively. The antennas, are vertically polarized. We assume that the far field functions are measured such that the antenna position, placement, and car body effects are taken into account, e.g., for side-windows mounted antennas, the blockage created by the vehicle body is accounted for in the far field functions.  Then,
the baseband channel gain between \gls{Tx} antenna $m$ and \gls{Rx} antenna $l$ can be modeled as~\cite[Ch.~6]{WirelessCom_Molisch} 
		\begin{align}
	h_{l,m}(t)=
	a(t)\gt_m(\phit,\theta^{\mathrm{s}})\gr_l(\phir,\theta^{\mathrm{r}})\mathrm{e}^{\jmath 	\omt_m}\mathrm{e}^{-\jmath 	\omr_l},\label{def:channel:v1} 
	\end{align}
	where $a(t)=|a(t)| \mathrm{e}^{-\jmath 2\pi f_{\text{c}}\tau_0(t)}$ is the complex-amplitude of the dominant component, $f_{\text{c}}$ is the carrier frequency, $\tau_0(t)=d_{0,0}(t)/c$, $d_{0,0}(t)$ is the physical path length between the transmit and receive reference antenna pair $(0,0)$, and $c$ is the speed of light in free space. The \gls{AOD} and \gls{AOA} in azimuth and elevation planes are denoted by $(\phit,\theta^{\mathrm{s}})$ and $(\phir,\theta^{\mathrm{r}})$, respectively.
	The relative phase shifts with respect to the reference antennas, $\omt_m$, $\omr_l$, depend on the \gls{AOD} and \gls{AOA} and they are given by~\cite[Ch.~6]{WirelessCom_Molisch} 
	\begin{align}
		\omt_m&=\langle\boldsymbol{k_\text{c}}(\phit,\theta^{\mathrm{s}}),\boldsymbol{u}^{\mathrm{s}}_m-\boldsymbol{u}^{\mathrm{s}}_0\rangle\label{def:channel:relativepahses:t},\\
		\omr_l&= \langle\boldsymbol{k_\text{c}}(\phir,\theta^{\mathrm{r}}),~\boldsymbol{u}^{\mathrm{r}}_l-\boldsymbol{u}^{\mathrm{r}}_0\rangle,\label{def:channel:relativepahses:r}
	\end{align}
	where $\langle \cdot\,,\cdot\rangle$ denotes the inner product, $\boldsymbol{u}^{\mathrm{s}}_m$ and  $\boldsymbol{u}^{\mathrm{r}}_l$ are the coordinates of $m^{\text{th}}$ transmit and $l^{\text{th}}$ receive antennas respectively, and $\boldsymbol{k_\text{c}}(.)$ is the unit wave vector in the direction of \gls{AOD} or \gls{AOA}, with coordinates $(k_{\text{x}},k_{\text{y}},k_{\text{z}})=2\pi/\lambda_{\text{c}}(\sin (\theta)\cos(\phi),\sin(\theta)\sin(\phi),\cos(\theta))$. The elevation plane angle $\theta$ is here defined as the angle between the z-axis and the vector of interest (i.e., polar angle). 
	
	We note that in~\eqref{def:channel:v1}, $a(t)$, $(\phit,\theta^{\mathrm{s}})$ and $(\phir,\theta^{\mathrm{r}})$ are assumed to be the same for all antenna pairs $(l,m)$. This assumption is reasonable when the distance between the \gls{Tx} and \gls{Rx}, or the distance separating the refection point and the antenna arrays (\gls{SBR} propagation), is much larger than the inter-separation between antenna elements of both, \gls{Tx} and \gls{Rx} arrays~\cite[Ch.~7]{book_LOS_MIMO}. Furthermore, the relative phase differences as expressed in~\eqref{def:channel:relativepahses:r} and~\eqref{def:channel:relativepahses:t} follows from the same assumption. Besides this, it is worth noting that the channel model is not restricted to a specific antenna array arrangement.
	
	Taking into account that the vehicles are relatively of the same height, the elevation angles are restricted to a narrow sector and it is therefore of less importance compared to the azimuth angles for a scenario of V2V communication. Following that, we assume that the dominant component is arriving along the azimuth plane with $\theta^{\mathrm{s}}\approx \pi/2$, $\theta^{\mathrm{r}}\approx \pi/2$. Consequently, the \gls{AOD} and \gls{AOA} can be restricted to the azimuth plane angles $\phit$ and $\phir$ in~\eqref{def:channel:v1} and the far field functions are expressed as
	\begin{align}\label{eq:antenna:phi:only}
	  \gt_m(\phit)\triangleq \gt_m(\phit,\pi/2),\quad \gr_l(\phir)\triangleq \gr_l(\phir,\pi/2).
	\end{align}

\subsection{Traffic Model of IEEE802.11p Cooperative Service}
Consider a scenario where vehicles periodically broadcast their status information to all vehicles in their proximity to create a cooperative environment. Such functionality can be supported by IEEE802.11p V2X technology, where every vehicle broadcast a \gls{CAM} every $T$~s that include information about their dynamic status, like position, speed, heading, etc. The period of dissemination is specified to be in the range of $0.1\leq T\leq 1$~s~\cite{CAM}. The physical layer supports multiple data rates, however, for \glspl{CAM} a data rate of $6$~Mbit/s is deemed suitable~\cite{IEEE11p_rate}. Given that a reasonable \gls{CAM} size is in the range $100$ to $500$ bytes~\cite{IEEE11p_rate}, the message duration $T_{\mathrm{m}} < 0.7$~ms. 

Motivated by the high dissemination frequency of \gls{CAM} messages, it has been suggested in several works including~\cite{Aoinfo}, to measure the reliability of a \gls{C-ITS} application that depends on their content using~\Gls{AoI}. 
In this framework, given a packet is generated by a vehicle at time $t=0$, then transmitted and correctly received at time $t=\tau$ by a receiving vehicle, if the time elapsing until the next packet reception exceeds a maximum allowable age $A_{\text{max}}$ then an outage is declared. Exceeding $A_{\text{max}}$ implies that the age of the status information decoded at $t=\tau$ is outdated and cannot be used by a \gls{C-ITS} application at the receiving vehicle. If the latency between the generation and reception of packets is neglected (i.e, $\tau\approx 0$), age-of-information is equivalent to \gls{IRT} which is defined as the time separating two successful reception of packets and it is elaborated in~\cite{PIR}. \Gls{BrEP} defined as the probability of losing $K$ consecutive \glspl{CAM} can be thought of as a physical layer counterpart to \gls{AoI} and it was first proposed in~\cite{ACN}. If latency is neglected, the two parameters can be related as $A_{\text{max}}=KT$ where $K$ is the burst length and $T$ is the period of dissemination. 

\begin{figure}[b]
	\centering
	\subfigure[Analog beamforming network of phase shifters (ABN)]{
		\includegraphics[width=0.95\columnwidth]{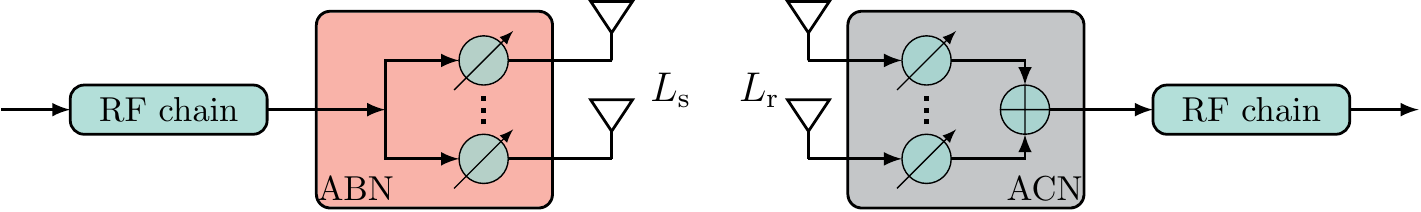}
		\label{fig:ABN}}
	\subfigure[Transmit antenna switching network (ASN)]{
		\centering
		\includegraphics[width=0.95\columnwidth]{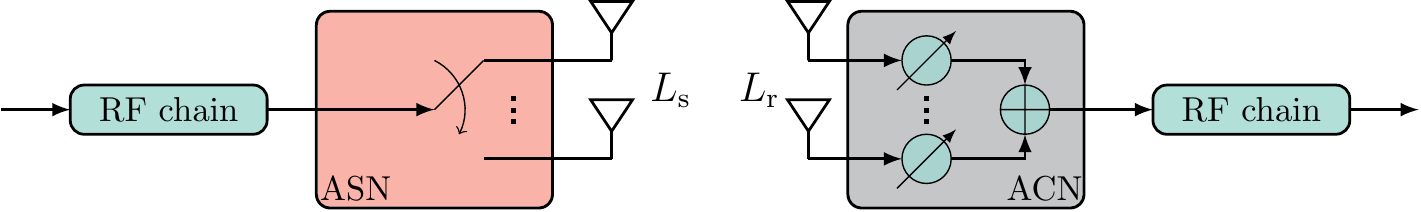}
		\label{fig:ASN}	
	}
	\caption{Transmission strategies structures.}
\end{figure}
\subsection{Transmission Strategies}
In the following, we assume that a reference transmitting vehicle is equipped with $\Lt$ antennas and
	the receiving vehicles are equipped with $\Lr$ antennas. 
The receiver structure is based on an \gls{ACN} as proposed in~\cite{ACN}. The combining network does not depend on \gls{CSI}, and it is composed of phase shifters modeled as affine functions of time according to $\varphi^{\mathrm{r}}_l=(\phr_l t+ \phrI_l)$, where $\phr_l \in \mathbb{R}$ represents a phase slope, $\phrI_l \in [0,2\pi)$ is an unknown initial phase offset, and $0\leq l\leq \Lr-1$.
We are interested in finding the optimal parameters of the network for two multiple antennas beamforming strategies. 
Before presenting the strategies we would like to develop simple generic equations for the received signal of a particular \gls{VU}. Let $\boldsymbol{b}$ be the beamforming vector, the received signal can be expressed as 
\begin{align}
\boldsymbol{r}= a(t)\boldsymbol{H}\boldsymbol{b} x(t)+ \boldsymbol{n},
\end{align}
where $x(t)=\tilde{x}(t-\tau_0(t))$, $\tilde{x}(t)$ is the transmitted baseband signal, and $\tau_0(t)$ is the propagation delay. 
The channel matrix $\boldsymbol{H}$ has entries $[\boldsymbol{H}]_{l,m}=h_{l,m}(t)/a(t)$ following~\eqref{def:channel:v1}.
The noise vector $\boldsymbol{n}$ has $\Lr$ elements, each modeled as independent white Gaussian noise over the system bandwidth with $\mathcal{CN}(0,\sigma^2_{\mathrm{n}})$. 
At the receiver, we apply the analog combining vector $\boldsymbol{w}$,
\begin{align}\label{eq:w:ACN}
&[\boldsymbol{w}]_l=\mathrm{e}^{-\jmath(\phr_l t+ \phrI_l)},\quad 0\leq l\leq \Lr-1.
\end{align}
The combined signal is given by
\begin{align}
r(t)&=a(t)x(t)\boldsymbol{w}^{\textsf{H}}\boldsymbol{H}\boldsymbol{b}  + \boldsymbol{w}^{\textsf{H}} \boldsymbol{n}\\
&=a(t)x(t)c(t) + \tilde{n}(t),\label{eq:r:generic}
\end{align}
where $c(t)=\boldsymbol{w}^{\textsf{H}}\boldsymbol{H}\boldsymbol{b}$ and $\tilde{n}(t)=\boldsymbol{w}^{\textsf{H}} \boldsymbol{n}$ denote the effective channel gain and noise at the output of the \gls{ACN}. The signal $\tilde{n}(t)$ is a zero-mean white Gaussian noise with $P_{\text{n}}=\mathbb{E}\{|\boldsymbol{w}^{\textsf{H}} \boldsymbol{n}|^2 \}= \Lr\sigma^2_{\mathrm{n}}$, and $P_{\mathrm{r}}=\mathbb{E}\{|a(t)x(t)|^2\}$ is the average received power. 
In the following, we present the two beamforming strategies and their corresponding effective channel gains. 
\subsubsection{\Gls{ABN} of Phase Shifters} 
We propose a transmitter structure that is similar to the receiver combining technique. Namely, we use phase shifters modeled as $\varphi^{\mathrm{s}}_m(t)=(\pht_m t+ \phtI_m)$ per antenna branch $0\leq m\leq\Lt-1$, where $\pht_m\in \mathbb{R}$ is a phase slope and $\phtI_{m}\in [0,2\pi)$ is an unknown initial phase offset. The overall structure is shown in \fig~\ref{fig:ABN}. 
The analog beamforming vector is given by 
\begin{align}\label{eq:b:abn}
&[\boldsymbol{b}]_m=\frac{1}{\sqrt{\Lt}}\mathrm{e}^{\jmath(\pht_m t+ \phtI_m)}, \quad 0\leq m\leq \Lt-1.
\end{align}
Note that by setting the phase slopes, the beamforming vector is determined ($\phtI_{m}$ is unknown and could take any value in $[0,2\pi)$).
The factor $1/\sqrt{\Lt}$, comes from splitting equally the power among the transmit antenna branches, which is a reasonable measure in absence of channel knowledge. Also, it ensures the use of similar power level compared to the single transmit antenna case.  Using~\eqref{def:channel:v1},~\eqref{eq:antenna:phi:only} and~\eqref{eq:r:generic} the corresponding effective channel gain is given by 
\begin{align}
c^{(\mathrm{a})}(t)=&\frac{1}{\sqrt{\Lt}} \sum_{m=0}^{\Lt-1}\gt_m(\phit)\mathrm{e}^{-\jmath(-\omt_m-\pht_mt-\phtI_m)}\nonumber \\ &\times\sum_{l=0}^{\Lr-1} \gr_l(\phir)\mathrm{e}^{-\jmath(\omr_l-\phr_lt-\phrI_l)}.\label{eq:c:ABN}
\end{align}
\subsubsection{Transmit \Gls{ASN}}
Instead of using analog phase shifters, we consider a transmitter that uses only switches. Since we are restricting any form of channel based control of the switches, an \gls{ASN} uses a single antenna element for a transmission then switch to the next element for the next transmission. 
The overall structure is illustrated in \fig~\ref{fig:ASN}.
Let $k$ denote the packet index, the \gls{ASN} beamforming vector for $kT\leq t\leq kT+T_{\mathrm{m}} $, can be expressed as 
\begin{align}
	&[\boldsymbol{b}]_m=1, \quad m=\rem(k,\Lt)\\
	&[\boldsymbol{b}]_i=0, \quad\forall i\neq m\nonumber 
\end{align}
where $\rem(a,b)$ denotes the remainder of dividing $a$ by $b$. We note that the same antenna element is used periodically every $\Lt$ transmissions.
The received signal following this strategy for the $k^{\text{th}}$ packet is given by~\eqref{eq:r:generic}, with effective channel gain at $kT\leq t\leq kT+T_{\mathrm{m}}$ 
\begin{align}\label{eq:c:ASN}
	c^{(\mathrm{b})}(t)=\gt_m(\phit)\mathrm{e}^{\jmath\omt_m}\sum_{l=0}^{\Lr-1} \gr_l(\phir)\mathrm{e}^{-\jmath(\omr_l-\phr_lt-\phrI_l)},
\end{align} 
where $m=\rem(k,\Lt)$.
Note that the power of the signal is not splitted in this case, and a single antenna element uses full power for every transmission. 

	\subsection{Assumptions}\label{Sec:sub:assumptions}
Before tackling the design task we consider a number of assumptions. 
First, since the packet duration is very small in comparison to the period, $T_{\mathrm{m}} \ll T$, we assume that the effective channel gains are constant over the packet duration for both strategies.
That is 
\begin{align}\label{eq:assumption:Ccte}
c^{(.)}(t)\approx c^{(.)}(kT),\quad kT\leq t\leq kT+T_{\mathrm{m}}.
\end{align}
Second, given the channel model, if the \gls{AOD}, \gls{AOA} coincides with directions along which the antenna systems have low gain, and the directions are sustained for $KT$~s then an outage may occur. Thus, taking into account this worst-case propagation scenario and as part of our robust design approach, we assume that the dominant component existing between the \gls{Tx} and \gls{Rx} experience negligible change over the duration of $KT$~s. Consequently, the following assumptions follow.
\begin{itemize}
	\item The \gls{AOD}, $\phit$, and \gls{AOA}, $\phir$, are assumed to experience negligible change and thus are modeled as constant over the duration of $KT$~s.
	\item From~\eqref{def:channel:relativepahses:r} and~\eqref{def:channel:relativepahses:t}, we see that $\omt_m$ and $\omr_l$ depend on the \gls{AOD}, \gls{AOA} and the geometry of the antenna arrays (which is fixed), therefore they can be assumed constant over the duration of $KT$~s as well.
	\item The average received power along the dominant component expressed as $P_{\mathrm{r}}=\mathbb{E}\{|a(t)x(t)|^2\}$ is assumed constant over $KT$~s. 
\end{itemize}

\section{Design of the \gls{Tx}-\gls{Rx} Schemes}\label{sec:design}
There exists several users to which \gls{CAM} messages need to reach. From a robustness aspect, we want to ensure certain performance for the worst receiving \gls{VU}. 
That depends on the applied beamforming vector $\boldsymbol{b}$, combining vector $\boldsymbol{w}$ and the channel between the transmitting and receiving \glspl{VU}. For our simplified model, 
the channel can be represented by the \gls{AOD}, \gls{AOA} and the far field functions of antennas. Assuming that all receiving \glspl{VU} have the same antenna system, (i.e., same number of antennas $\Lr$ and same far field functions $\gr_l$), the worst receiving user is defined by the worst-case \gls{AOD}, \gls{AOA} for a given $\boldsymbol{b}$ and $\boldsymbol{w}$.  
We will find out later that the solution found under this assumption, is also optimal when the system is generalized to receiving users with different number of antennas and different far field functions. 
	
We did set our framework to ensure the robustness in the system. To quantify performance, we use the concept of \gls{BrEP}, that is, the probability of losing $K$ consecutive \gls{CAM} packets. To be able to derive some analytical results, we resort to the assumption that packet error probability is an exponential function of \gls{SNR}, and that packet errors are statistically independent. As a consequence, the minimization of \gls{BrEP} is equivalent to the maximization of the sum of the \gls{SNR} of the $K$ packets, referred to as sum-\gls{SNR} in this work. Exact details of relating these two parameters can be found in~\cite[\secR~III]{ACN},~\cite[\secR~III.B]{HC}.

Thus, under the above introduced assumptions, the design objective, which can be formulated as minimizing \gls{BrEP} of $K$ packets for the worst receiving \gls{VU}, is equivalent to maximizing the sum-\gls{SNR} of $K$ packets for the worst-case \gls{AOA} and \gls{AOD}. 
From here, we go through the design procedures for the two schemes separately. 
\begin{figure*}[t]
	\centering
\begin{align}
\JA(\phir,\phit,&\phrV,\phtV,\psrV,\pstV )= 
\sum_{l=0}^{\Lr-2}\sum_{i=l+1}^{\Lr-1}c_{l,i}(\phir,\phit)\sum_{k=0}^{K-1}\cos\big(\Delta\psr_{l,i}-\Delta\phr_{l,i}kT\big)\nonumber\\
&+\sum_{m=0}^{\Lt-2}\sum_{j=m+1}^{\Lt-1}d_{m,j}(\phir,\phit)\sum_{k=0}^{K-1}\cos\big(\Delta\pst_{m,j}-\Delta\pht_{m,j}kT\big)\nonumber\\
&+\sum_{l=0}^{\Lr-2}\sum_{i=l+1}^{\Lr-1}\sum_{m=0}^{\Lt-2}\sum_{j=m+1}^{\Lt-1}c'_{l,i}(\phir)d'_{m,j}(\phit)
\sum_{k=0}^{K-1}\cos\big(\Delta\psr_{l,i}-\Delta\phr_{l,i}kT\big)\cos\big(\Delta\pst_{m,j}-\Delta\pht_{m,j}kT\big). \label{def:J:original}\tag{25}\\
c_{l,i}(\phir,\phit)&= \bigg(\sum_{m=0}^{\Lt-1}\frac{|\gt_m(\phit)|^2}{\Lt} \bigg) \overbrace{2\frac{|\gr_l(\phir)||\gr_i(\phir)|}{\Lr}}^{c'_{l,i}(\phir)}, \qquad
d_{m,j}(\phir,\phit)=\bigg(\sum_{l=0}^{\Lr-1}\frac{|\gr_l(\phir)|^2}{\Lr}\bigg) \overbrace{2\frac{|\gt_m(\phit)||\gt_j(\phit)|}{\Lt}}^{d'_{m,j}(\phit)}\label{def:J:cli:dmj}.\tag{26}\\
\cline{1-2}\nonumber
\end{align}
\end{figure*}
\subsection{Design of the Analog Beamforming Network}

\subsubsection{Deriving the \gls{SNR} per Packet}
Given an \gls{ABN} transmitter, the received signal is modeled by~\eqref{eq:r:generic}, where the effective channel gain is given by $c^{(\mathrm{a})}$~\eqref{eq:c:ABN}. Incorporating the assumptions in \secR~\ref{Sec:sub:assumptions}, $c^{(\mathrm{a})}$ can be approximated when $kT\leq t \leq kT+T_{\mathrm{m}}$ as
\begin{align}
	c^{(\mathrm{a})}(kT)&=\frac{1}{\sqrt{\Lt}} \sum_{m=0}^{\Lt-1}|\gt_m(\phit)|\mathrm{e}^{-\jmath(\pst_m-\pht_mkT)}\nonumber
	\\ &\times\sum_{l=0}^{\Lr-1} |\gr_l(\phir)|\mathrm{e}^{-\jmath(\psi_l-\phr_lkT)}, \label{eq:ca}
\end{align}
where the effective channel phases are given by
\begin{align}
	\pst_m&=\rem{(-\omt_m-\phase{\gt_m(\phit)}-\phtI_m, 2\pi)}\label{def:psit},\\
	\psr_l&=\rem {(\omr_l-\phase{\gr_l(\phir)}-\phrI_l, 2\pi)},\label{def:psir}
\end{align}
and they are approximately constant over $KT$~s.
Let $\psrV$, $\phrV$ be $\Lr$-vectors with entries corresponding to $\psr_l$ and $\phr_l$, respectively, 
and analogously are defined the $\Lt$-vectors $\pstV$ and $\phtV$.
Then, recalling~\eqref{eq:r:generic} we can express the $k^{\text{th}}$ packet \gls{SNR}, for  $k=0, 1, \cdots, K-1$ as
\begin{align}\label{def:gamma:txrx}
\snrA_k=&\frac{P_{\mathrm{r}}} {\Lr\sigma^2_{\mathrm{n}}}|c^{(\mathrm{a})}(kT)|^2\\
=&\frac{P_{\mathrm{r}}}{\sigma^2_{\mathrm{n}}}  \frac{\Gr(\phir, \phrV,\psrV,k)}{\Lr} \frac{\Gt (\phit, \phtV,\pstV,k)}{\Lt}, 
\end{align}
where
\begin{align}\label{def:Gr}
\Gr(\phir, \phrV,\psrV,k)=&\bigg|\sum_{l=0}^{\Lr-1} |\gr_l(\phir)|\mathrm{e}^{-\jmath(\psr_l-\phr_lkT)} \bigg|^2\\
=\sum_{l=0}^{\Lr-1}|\gr_l(\phir)|^2&
+2\sum_{l=0}^{\Lr-2}\sum_{i=l+1}^{\Lr-1}|\gr_l(\phir)||\gr_i(\phir)|  \nonumber \\
&\times\cos\big(\psr_i-\psr_l-(\phr_i-\phr_l)kT\big).\label{eq:Gr:elaborated}
\end{align}
and 
\begin{align}
\Gt (\phit, \phtV,\pstV,k)=\bigg|\sum_{m=0}^{\Lt-1} |\gt_m(\phit)|\mathrm{e}^{-\jmath(\pst_l-\pht_lkT)}\bigg|^2.
\end{align}
From \secR~\ref{Sec:sub:assumptions}, we know that $P_{\mathrm{r}}$ is assumed approximately constant over burst duration.
Thus, we normalize the \gls{SNR} with respect to $P_{\mathrm{r}}/\sigma^2_{\mathrm{n}}$ and we sum over a burst of $K$ consecutive packets 
to obtain the normalized sum-\gls{SNR} given by
\begin{align}
	\SsnrA&(\phir,\phit,\phrV,\phtV,\psrV,\pstV )=\sigma^2_{\mathrm{n}}/P_{\mathrm{r}}\sum_{k=0}^{K-1}\snrA_k\nonumber\\ 
&=K \overline{G}(\phir,\phit)+ \JA(\phir,\phit,\phrV,\phtV,\psrV,\pstV )\label{def:S:GplusJ}.
\end{align}
That is, the sum-\gls{SNR} $\SsnrA(.)$ has a part that depends only on the \gls{AOA} and \gls{AOD},
 \begin{align}\label{def:G}
\overline{G}(\phir,\phit)=\sum_{l=0}^{\Lr-1}\frac{|\gr_l(\phir)|^2}{\Lr} \sum_{m=0}^{\Lt-1}\frac{|\gt_m(\phit)|^2}{\Lt},
\end{align}
and a part that captures the channel variation $\JA(\cdot)$, which, with some simplification of notation, and introduction of the operator 
\begin{align}\label{def:operator:Delta}
	\Delta x_{l,i}\triangleq x_i-x_l, 
\end{align}\addtocounter{equation}{2}
 can be expressed according to~\eqref{def:J:original}--\eqref{def:J:cli:dmj} on top of the page.

\subsubsection{Optimization Problem}
We are interested in finding the design parameters $\phrV,\phtV$ that yields a robust system, i.e., that maximizes the sum-\gls{SNR} for worst-case \gls{AOA}, \gls{AOD}. Besides that, we need to account for the worst-case effective channel phases $(\psrV,\pstV)$ which depend on the initial unknown offset of phase shifters ($\phrI_l,\phtI_m$) at both \gls{Rx} and \gls{Tx}. 
Thus, the problem can be stated formally as 
\begin{align}\label{def:opt:problem:ABN}
	&(\phrSV,\phtSV)=
	\arg \sup_{(\phrV,\phtV)} \inf_{\substack{(\phir,\phit)\\(\psrV,\pstV )}} \SsnrA(\phir,\phit,\phrV,\phtV,\psrV,\pstV ),%
\end{align}
where,  $\phrV \in \mathbb{R}^{\Lr}$, $\phtV \in \mathbb{R}^{\Lt}$, $
\psrV \in [0,2\pi)^{\Lr}$, $\pstV \in [0,2\pi)^{\Lt}$ and $
\phir,\phit \in [0,2\pi)$.
The solutions to this problem can be deduced from the following theorem.
\begin{theorem}\label{theo:main}
Given $\SsnrA$ defined in~\eqref{def:S:GplusJ} and $\Lr+\Lt>2$, let 
	\begin{align}
		\SsnrA^\star(\phir,\phit)&\triangleq \sup_{(\phrV,\phtV)} \inf_{(\psrV,\pstV )}\SsnrA(\phir,\phit,\phrV,\phtV,\psrV,\pstV )\nonumber.
	\end{align}
	Then, for any \gls{AOA}, \gls{AOD}, we have
	\begin{enumerate}
	\item The objective function is upper bounded as
	\begin{align}
	\SsnrA^\star(\phir,\phit) \leq K \overline{G}(\phir,\phit).\label{eq:theo:1:UB}
	\end{align}
	    \item If $\Lr\Lt\leq K$, the upper bound is achievable
	    \begin{align}
	        \SsnrA^\star(\phir,\phit)=K \overline{G}(\phir,\phit).\label{eq:theo:1}
	    \end{align}
	   \item  A solution $(\phrSV,\phtSV)$ yields~\eqref{eq:theo:1}, when $\Lr\Lt\leq K$ if 
	     \begin{subequations}\label{eq:th:optCon}
	\begin{align}
	\Delta \phrS_{l,i}T/2&\in \mathcal{X^*}, \label{eq:th:optCon:x}\\
	\Delta \phtS_{m,j}T/2 &\in \mathcal{X^*}, \label{eq:th:optCon:z}\\
	(\Delta \phrS_{l,i}\pm \Delta \phtS_{m,j})T/2 &\in \mathcal{X^*}, \label{eq:th:optCon:xz}
	\end{align} 
\end{subequations}
 where $\mathcal{X^*}=\{q\pi/K:q\in\mathbb{Z}\}\setminus\{n\pi:n\in \mathbb{Z}\}$, $0\leq l<i\leq \Lr-1$ and $0\leq m<j\leq \Lt-1$. 
The condition is restricted to~\eqref{eq:th:optCon:z} or~\eqref{eq:th:optCon:x}, when $\Lr=1$ or $\Lt=1$, respectively. 
	\item If $|\gr_l(\phir)|>0$, $\forall l$ and $|\gt_m(\phit)|>0$, $\forall m$, then
~\eqref{eq:theo:1} $\implies$~\eqref{eq:th:optCon} and $\Lr\Lt\leq K$.
	\end{enumerate}
\end{theorem}
\begin{proof}
	See Appendix~\ref{app:theo:main}.
\end{proof}
\begin{remark}\label{}
The solutions satisfying~\eqref{eq:th:optCon} are independent of $\phir$ and $\phit$, and thus they solve~\eqref{def:opt:problem:ABN} when $\Lr\Lt\leq K$.
\end{remark}
\begin{remark}\label{remark:theorem:main}
An optimal solution $(\phrSV,\phtSV)$ (i.e., satisfies~\eqref{eq:th:optCon}) achieve for $\psrV \in [0,2\pi)^{\Lr}$, $\pstV \in [0,2\pi)^{\Lt}$,
\begin{align}\label{eq:remark}
    \SsnrA(\phir,\phit,\phrSV,\phtSV,\psrV,\pstV )=K \overline{G}(\phir,\phit).
\end{align}
\end{remark}
\begin{proof}
	See Appendix~\ref{app:theo:main}.
\end{proof}

Theorem~\ref{theo:main} indicates that choosing a set of phase slopes satisfying~\eqref{eq:th:optCon} maximizes the sum-\gls{SNR} for the worst-case effective channel phase vectors, for all \gls{AOA}, \gls{AOD} including the worst-case directions. Thus,~\eqref{eq:th:optCon} are solutions to~\eqref{def:opt:problem:ABN}, when $\Lr\Lt\leq K$. Assuming an antenna system with $|\gr_l(\phir)|>0$, $\forall l$ and $|\gt_m(\phit)|>0$, $\forall m$, then we know that the only way to solve~\eqref{def:opt:problem:ABN} is to use a set of phase slopes satisfying~\eqref{eq:th:optCon}. That implies that~\eqref{eq:th:optCon} is sufficient and necessary optimality condition in this case. Note that the last assumption on antenna system is easily satisfied for physical antennas which typically radiates in all directions including nulls. These correspond usually to very low, but non-zero gains (several dBs below zero).
Recalling~\eqref{def:S:GplusJ}, we see that Remark~\ref{remark:theorem:main} points out that the optimal phase slopes average out the variation of sum-\gls{SNR} due to effective channel phases, i.e., $\JA(\phir,\phit,\phrSV,\phtSV,\psrV,\pstV)=0,$ $\forall~ \pstV, \psrV$. Since the optimal sum-\gls{SNR} is proportional to $\overline{G}(\phir,\phit)$, we refer to $\overline{G}$ as the effective radiation pattern realized using \gls{ABN}.   

The theorem results are general and they apply to the special cases $\Lt=1$ or $\Lr=1$. The optimality conditions boils down to $\Lr\leq K$ and~\eqref{eq:th:optCon:x} for the former special case, and to $\Lt\leq K$ and \eqref{eq:th:optCon:z} for the latter. The system with $\Lt=1$ has already been studied in~\cite{ACN}. The derived condition~\eqref{eq:th:optCon:x} coincides with optimality condition of phase slopes obtained in~\cite[Theorem~1, eq. (18)]{ACN}. However, in that work, it was shown that the phase slopes satisfying~\eqref{eq:th:optCon:x} ensures a lower bound on the objective function when $\Lr\leq K$, while their optimality was shown to hold only in the special cases of $\Lr\in \{2, 3\}$. Here, Theorem~\ref{theo:main} extends the optimality of phase slopes satisfying~\eqref{eq:th:optCon:x} to any system with $\Lr\leq K$ ($\Lt=1$).

Given the conditions~\eqref{eq:th:optCon} and $\Lr\Lt\leq K$ we can derive phase slopes constructions that achieves optimality. These are stated in the following corollary.
\begin{corollary}\label{cor:ABN:ph_slopes}
	If $\Lr\Lt\leq K$, then the following phase slopes
		\begin{align}
		\phrS_l&= l\frac{2\pi}{KT}, ~\phtS_m=m \Lr \frac{2\pi}{KT},\label{eq:th:ph_slopes}\\
		\phrS_l&= l\Lt \frac{2\pi}{KT},~ \phtS_m=m \frac{2\pi}{KT}, \label{eq:th:ph_slopes:reciprocal} 	
		\end{align}
	where $0\leq l\leq \Lr-1$ and $0\leq m\leq \Lt-1$, satisfy~\eqref{eq:th:optCon} and thus are optimal.
\end{corollary}
\begin{proof}
	See Appendix~\ref{app:theo:main} (Lemma~\ref{lem:LrLt:K}).
\end{proof}
These two phase slopes constructions are not unique. 
The construction~\eqref{eq:th:ph_slopes} yields the same receiver phase slopes that were suggested in~\cite[Theorem~1, eq. (19)]{ACN}. 
These constructions require the knowledge of the number of antennas at the receiver or the transmitter. So far, we assumed that all receiving \glspl{VU} have the same number of antennas. However, we can generalize this solution to \glspl{VU} with different number antennas.

\subsubsection{Supporting \glspl{VU} with Different Number of Antennas}
Instead of assuming that all users have same number of antennas, let us assume that the maximum number of antennas any user can have is $\Lrm$ at the receiver and $\Ltm$ at the transmitter. Moreover, we assume that all users are aware of these two parameters and they satisfy $\Ltm\Lrm\leq K$.
Following that we can use~\eqref{eq:th:optCon} and attempt to find phase slopes that works for any $\Lt \leq \Ltm$ and $\Lr\leq \Lrm $.  The following corollary give us such a solution.
\begin{corollary}\label{cor:ABN:ph_slopes:multi_users}
	The phase slopes construction given by
		\begin{align}
		\phrS_l&= l\frac{2\pi}{KT}, \quad \phtS_m=m  \Lrm \frac{2\pi}{KT},\label{eq:th:ph_slopes:multiVUs}\\
		\phrS_l&= l\Ltm\frac{2\pi}{KT}, \quad \phtS_m=m  \frac{2\pi}{KT},\label{eq:th:ph_slopes:reciprocal:multiVUs}
		\end{align}
	where $0 \leq l\leq \Lr-1$, $0\leq m\leq \Lt-1$ satisfies~\eqref{eq:th:optCon} when $\Lrm\Ltm\leq K$, and $\Lr\leq \Lrm$, $\Lt\leq \Ltm$.
\end{corollary}
\begin{proof}
	This can be reached straightforwardly by substituting in~\eqref{eq:th:optCon} and using the fact that $\Lrm\Ltm\leq K$ (For exact details see proof of Corollary~\ref{cor:ABN:ph_slopes}). 
\end{proof}
Note that in such case a transmitter does not need to know the number of antennas at the receiver. However, a knowledge of the maximum number of antennas that a user can have is required.

The phase slopes at the \gls{Tx} does not depend on the \gls{AOA}, \gls{AOD} nor on the far field functions of antennas. This implies that a transmitter can use an \gls{ABN} beamforming vector with fixed phase slopes to achieve an optimal performance for any \gls{VU} including the worst one (i.e., including the user with the worst sum-\gls{SNR}). Hence, the solutions to~\eqref{def:opt:problem:ABN} solves the problem of maximizing the sum-\gls{SNR} for the worst receiving user, also when \glspl{VU} have different number of antennas $\Lr\leq\Lrm$ and different far field functions.

\subsection{Design of the Transmit Antenna Switching Network}

\subsubsection{Deriving the \gls{SNR} per Packet}
We follow the same steps as done when deriving the SNR for the \gls{ABN} scheme. Given the received signal~\eqref{eq:r:generic}, with effective channel gain $c^{\mathrm{(b)}}$~\eqref{eq:c:ASN}, and adopting the assumptions in \secR~\ref{Sec:sub:assumptions} the SNR for the $k^{\text{th}}$ packet, $k=0, 1, \cdots, K-1$ can be expressed as
\begin{align}\label{eq:snr_k:ASN}
\snrB_k= \frac{P_{\mathrm{r}}}{\Lr\sigma^2_{\mathrm{n}}}|c^{\mathrm{(b)}}(kT)|^2= \frac{P_{\mathrm{r}}}{\sigma^2_{\mathrm{n}}}|\gt_m(\phit)|^2 \frac{\Gr(\phir,\phrV,\psrV,k)}{\Lr},
\end{align}
where $m=\rem(k,\Lt)$, $\Gr(.)$ is given by~\eqref{def:Gr} and $\psrV$ is an $\Lr$-vector with elements $\psr_l$ defined in~\eqref{def:psir}.
\gls{ASN} switches antenna after each transmission. To simplify its analysis, we assume that $K/\Lt=\Kr$ is an integer, implying that each antenna element is used to transmit an equal number of $\Kr$ packets within a burst of $K$ packets. Then, the \gls{ASN} normalized sum-\gls{SNR} is given by
\begin{align}
\SsnrB&(\phir,\phit,\phrV,\psrV )=\sigma^2_{\mathrm{n}}/P_{\mathrm{r}} \sum_{k=0}^{K-1} \snrB_k \label{eq:Sb:general:definition}\\
&=\sum_{m=0}^{\Lt-1} |\gt_m(\phit)|^2 \sum_{k'=0}^{\Kr-1} \frac{\Gr(\phir,\phrV,\psrV,m+k'\Lt)}{\Lr},\label{eq:SNR:ASN}
\end{align}
	From the expression we observe that a packet is sent using the $m^{\text{th}}$ antenna, periodically, every $\Lt$ transmissions.
Using~\eqref{eq:Gr:elaborated}, the normalized sum-\gls{SNR} can be elaborated and stated as
\begin{align}
\SsnrB(\phir,\phit,\phrV,&\psrV )=
K \overline{G}(\phir,\phit) +
\JB(\phir,\phit,\phrV,\psrV ),\label{def:S:SA}
\end{align}
where $\overline{G}(\phir,\phit)$ is given by~\eqref{def:G} and 
\begin{align}
&\JB(\phir,\phit,\phrV,\psrV )=\sum_{m=0}^{\Lt-1} |\gt_m(\phit)|^2\sum_{l=0}^{\Lr-2}\sum_{i=l+1}^{\Lr-1} c'_{l,i}(\phir) \nonumber\\ &\times\sum_{k'=0}^{\Kr-1}\cos\big(\Delta\psr_{l,i}-\Delta\phr_{l,i}(m+k'\Lt)T\big),\label{def:J:SA}
\end{align}
where $c'_{l,i}(\phir)$ is defined in~\eqref{def:J:cli:dmj}.
\subsubsection{Optimization} 
For the \gls{ASN} only the receiver phase slopes are to be found. The optimization problem has a similar form to~\eqref{def:opt:problem:ABN}, and it is given by
\begin{align}\label{def:optPr:SA}
\phrSV= \arg \sup_{\phrV\in \mathbb{R}^{\Lr}} ~\inf_{ \substack{(\phit,\phir) \\\psrV \in [0,2\pi)^{\Lr}}} \SsnrB(\phir,\phit,\phrV,\psrV ).
\end{align}
Note that \gls{ASN} inherently treats all receiving \glspl{VU} the same. Thus, in spite of how many antennas and far field functions, different receiving \glspl{VU} have, the optimization problem~\eqref{def:optPr:SA} and its solutions applies to any user including the one experiencing worst sum-\gls{SNR}. 
The solution to the problem can be deduced from the following theorem.
\begin{theorem}\label{th:SA}
	Let $\SsnrB$ be as defined in~\eqref{def:S:SA}, $K/\Lt=\Kr \in \mathbb{Z}$, $\Lr>1$, and let 
	\begin{align}\label{def:Sbstar}
	\SsnrB^\star(\phir,\phit)&\triangleq \sup_{\phrV} \inf_{\psrV}\SsnrB(\phir,\phit,\phrV,\psrV).
	\end{align}
	Then, for any $(\phir,\phit)$
	\begin{enumerate}
	    \item The function is bounded as 
	    \begin{align}
	        \SsnrB^\star(\phir,\phit)\leq  K \overline{G}(\phir,\phit).\label{eq:theoremASN:claimUB}
	    \end{align}
	    \item If $\Lr\leq \Kr$ the bound is achievable, i.e.,
	     \begin{align}
	        \SsnrB^\star(\phir,\phit)= K \overline{G}(\phir,\phit),\label{eq:theoremASN:claim1}
	    \end{align}
	    with solutions that satisfy 
	\begin{align}\label{eq:th:SA:OptCon}
	\Lt \Delta\phrS_{l,i}T/2 \in \Xstb,
	\end{align} 
where $\Xstb=\{q\pi/\Kr:q\in\mathbb{Z}\}\setminus\{n\pi:n\in \mathbb{Z}\}$, and $0\le l<i\le \Lr-1$. 

\item  Assuming $|\gr_l(\phir)|>0$, $\forall l$, $|\gt_m(\phit)|>0$, $\forall m$, and $|\gt_m(\phit)|\neq C$, $\forall m$,
	then~\eqref{eq:theoremASN:claim1} $\implies$~\eqref{eq:th:SA:OptCon} and $\Lr\leq \Kr$.
	    
	    \item 	One optimal solution is given by
	\begin{align}\label{eq:th:SA:ph_slopes}
	\phrS_l&= l\frac{2\pi}{KT},\quad l=0,1,\cdots \Lr-1,~ \Lr\leq \Kr.
	\end{align}
	\end{enumerate}
\end{theorem}
\begin{proof}
	See Appendix~\ref{app:theo:SA}.
\end{proof}
Similarly to the \gls{ABN} case, the solutions in~\eqref{eq:th:SA:OptCon} are independent of the signal directions $(\phir,\phit)$, hence~\eqref{eq:th:SA:OptCon} solves~\eqref{def:optPr:SA} when $\Lr\leq \Kr=K/\Lt$.  
Moreover, under the assumptions indicated in Theorem~\ref{th:SA}~\textit{(iii)},~\eqref{eq:th:SA:OptCon} is sufficient and necessary optimality condition for~\eqref{def:optPr:SA}. Note that, $|\gt_m(\phit)|= C$, $\forall m$, is a special case where the antenna system is equivalent to $1\times \Lr$ system, and thus it is covered by Theorem~\ref{theo:main}. In particular, for $|\gt_m(\phit)|= C>0$, $\forall m$, and $|\gr_l(\phir)|>0$, $\forall l$,~\eqref{eq:theoremASN:claim1} $\implies$~\eqref{eq:th:optCon:x} and $\Lr\leq K$. 
We observe that, the optimality condition related to the number of antennas $\Lr\leq \Kr=K/\Lt$ is equivalent to what have been obtained for the \gls{ABN} scheme ($\Lr\Lt\leq K$). The phase slopes optimality condition~\eqref{eq:th:SA:OptCon} is, on the other hand, simpler than~\eqref{eq:th:optCon}.

We can draw two main conclusions from the theorem. First, the performance achieved using an \gls{ASN} is identical to that of an \gls{ABN}, $\SsnrB^\star(\phir,\phit)=\SsnrA^\star(\phir,\phit)$. Hence, the main difference between the two is the implementation. Second, the suggested optimal phase slopes construction~\eqref{eq:th:SA:ph_slopes} at the receiver side when an \gls{ASN} is used, coincides with the phase slopes construction~\eqref{eq:th:ph_slopes} (or ~\eqref{eq:th:ph_slopes:multiVUs}), which is used to combine the antenna signals when an \gls{ABN} is used at the transmitter. Hence, we can design the receiver \gls{ACN} such that it combines optimally both, signals arriving from \gls{ABN} and \gls{ASN} transmitters. 
Such a receiver does not require any knowledge of the number of antennas used at the transmitter side.

An important aspect of Theorem~\ref{th:SA} is the assumption $\Kr=K/\Lt \in \mathbb{Z}$, which can be easily met if all \gls{VU} use the same number of transmit antennas, however, in a context where \glspl{VU} are equipped with different number of transmit antennas it may not be easily met. In the following, we would like to see the implications of the assumption in such context.

\subsubsection{\glspl{VU} with Different Number of Transmit Antennas}\label{section:sub:ASN:MA}

Assume that the maximum number of antennas that can be used by any \gls{VU} to transmit and receive are $\Ltm$ and $\Lrm$, respectively. Further, assume that  $K/\Ltm\in \mathbb{Z}$, and $\Lrm\leq K/\Ltm$. 
	Receiving \glspl{VU} employ a phase slope vector $\phrSV$ that satisfies~\eqref{eq:th:SA:ph_slopes}. This selection of phase slopes is optimal for any \gls{Tx} with $\Lt$ satisfying $K/\Lt \in \mathbb{Z}$ (including $\Lt=\Ltm$). The optimal sum-\gls{SNR} is given by~\eqref{eq:theoremASN:claim1}. For transmitting \glspl{VU} with $K/\Lt \notin\mathbb{Z}$, the phase slopes are not known to achieve the optimal sum-\gls{SNR}. The performance is governed in this case by 
	\begin{align}\label{eq:Sstar:vir:counter}
	\inf_{\psrV \in [0,2\pi)^{\Lr}} \SsnrB(\phir,\phit,\phrSV,\psrV),
	\end{align}
	where $\SsnrB$ is as defined in~\eqref{eq:Sb:general:definition} (Note that~\eqref{eq:SNR:ASN}---~\eqref{def:J:SA} hold when $K/\Lt \in \mathbb{Z}$). No analytical solution to~\eqref{eq:Sstar:vir:counter} is available. However, a numerical characterization of the expression and a comparison with the optimal sum-\gls{SNR} attained by \gls{ABN} will be shown in the numerical results section.
\subsection {\gls{ABN} and \gls{ASN} Transceivers }
In spite of the advantage of \gls{ABN} over \gls{ASN} in supporting transmission with different number of antennas, we have shown that the two structures yield the same performance in the general case. However, each structure puts different requirements on the transceiver. In particular, an \gls{ABN}-\gls{ACN} transceiver needs to support two different sets of phase shifters slopes, one tuned for transmission ($\{\phtS_m\}$) and another for reception ($\{\phrS_l\}$). On the other hand, an \gls{ASN}-\gls{ACN} transceiver has to  support one set of phase shifters slopes  ($\{\phrS_l\}$), and be capable of transmitting with $\Lt$ times higher power than the power transmitted on each antenna branch using \gls{ABN}.
 From these requirements stem the main implementation trade-offs between the two structures, and depending on how transceivers with such requirements are implemented, \gls{ABN} and \gls{ASN} may differ in cost and complexity.  

\subsection{Alamouti Scheme Performance}
In the following, we would like to quantify the sum-\gls{SNR} of an Alamouti diversity scheme within context of this work.
Assume that the receiver uses an \gls{ACN}, while the transmitter with $\Lt=2$ applies an Alamouti encoding in space-time domain, in accordance with an \gls{OFDM} signal (see e.g.,~\cite{AlamoutiSTBC}, \cite[Ch.~22]{STBC_IEEE})\footnote{Alamouti can be applied to OFDM in space-frequency domain as well (see e.g.,~\cite{AlamoutiSFBC}), and it is found to yield the same result.}.
Assume that packets are composed of $N_{\textrm{sym}}$ \gls{OFDM} symbols, $\boldsymbol{s}_i$, each composed of $N$ subcarriers. 
The Alamouti encoding matrix is given by 
\begin{align}
\begin{bmatrix}
	\boldsymbol{s}_0  &\boldsymbol{s}_1\\
-\boldsymbol{s}_1^*  &\boldsymbol{s}_0^*
\end{bmatrix}\begin{array}{ll}
\rightarrow &\text{space}\\
\downarrow &\text{time}
\end{array}
\end{align}
where
$\boldsymbol{s}_0$, $\boldsymbol{s}_1$ are two consecutive \gls{OFDM} symbols. After applying the space-time mapping, symbols are converted to time domain, appended a cyclic prefix, then sent over the channel. Given the channel gain $h_{l,m}$~\eqref{def:channel:v1}, let $\boldsymbol{H}_m=[h_{0,m}, h_{1,m}, \cdots, h_{\Lr-1,m}]^\textsf{T}/a(t)$ denotes the $m^{\text{th}}$ column of $\boldsymbol{H}$, where $a(t)=|a(t)|\mathrm{e}^{-\jmath 2\pi f_{\text{c}} \tau_{0}(t)}$, then the received signal after \gls{ACN} combining with $\boldsymbol{w}$~\eqref{eq:w:ACN}, can be modeled as
\begin{align}\label{eq:r:Alamouti}
	r(t)=\frac{1}{\sqrt{2}}a(t)\boldsymbol{w}^{\textsf{H}}(\bar{x}_0(t) \boldsymbol{H}_0+\bar{x}_1(t) \boldsymbol{H}_1)+\boldsymbol{w}^{\textsf{H}}\boldsymbol{n},
\end{align}
where $\bar{x}_0(t)$ and $\bar{x}_1(t)$ 
are the transmitted Alamouti encoded signals delayed by $\tau_0(t)$, 
 and $\boldsymbol{w}^{\textsf{H}}\boldsymbol{n}$ is a zero-mean white Gaussian noise with variance $\mathbb{E}\{|\boldsymbol{w}^{\textsf{H}} \boldsymbol{n}|^2 \}= \Lr\sigma^2_{\mathrm{n}}$.

Given the assumptions in \secR~\ref{Sec:sub:assumptions} we can approximate $\bar{c}_m=\boldsymbol{w}^{\textsf{H}}\boldsymbol{H}_m/\sqrt{2}$, $m\in\{0,1\}$, for the $k^\text{th}$ packet, as
\begin{align}
    \frac{1}{\sqrt{2}}\gt_m(\phit)\mathrm{e}^{\jmath \omt_m} \sum_{l=0}^{\Lr-1} |\gr_l(\phir)| \mathrm{e}^{-\jmath(\psr_l-\phr_l kT)}.
\end{align}
where $\psr_l$ is given by~\eqref{def:psir}.
To decode the message, the receiver uses two consecutive \gls{OFDM} symbols (after discarding the cyclic prefix and conversion to the frequency domain), as follows
\begin{align}
	\boldsymbol{y}_{0}&=\bar{c}_{0}\boldsymbol{D}_{\textrm{a}}\boldsymbol{s}_{0} +\bar{c}_{1} \boldsymbol{D}_{\textrm{a}}\boldsymbol{s}_{1}+\boldsymbol{z}_0,\\
	\boldsymbol{y}_{1}&=-\bar{c}_{0}\boldsymbol{D}_{\textrm{a}}\boldsymbol{s}_{1}^*+\bar{c}_{1}\boldsymbol{D}_{\textrm{a}}\boldsymbol{s}_{0}^*+\boldsymbol{z}_{1},
\end{align}
where $\boldsymbol{D}_{\textrm{a}}=\bar{\boldsymbol{D}}_{\textrm{a}}^{(i)}$, is a diagonal matrix carrying the frequency response of the sampled finite channel impulse response associated with $a(t)$, at symbol duration $i$, and it is assumed constant over two consecutive \gls{OFDM} symbol durations (follows from the basic assumption of Alamouti scheme). The vectors $\boldsymbol{z}_0$, $\boldsymbol{z}_1$, are zero-mean independent white Gaussian noises with variance $\mathbb{E}\{|\boldsymbol{z}_i|^2\}=\Lr\sigma^2_{\mathrm{n}}/N_{\textrm{sym}}$, $i=1,2$.
Solving the equations following $(\bar{c}_{0}^*\boldsymbol{D}_{\textrm{a}}^*\boldsymbol{y}_{0}+\bar{c}_{1}\boldsymbol{D}_{\textrm{a}}	\boldsymbol{y}_{1}^*)$ and $(\bar{c}_{1}^*\boldsymbol{D}_{\textrm{a}}^*\boldsymbol{y}_{0}-\bar{c}_{0}\boldsymbol{D}_{\textrm{a}}	\boldsymbol{y}_{1}^*)$, and employing~\eqref{eq:r:Alamouti} we can deduce the normalized \gls{SNR} of the $k^{\text{th}}$ packet 
\begin{align}\label{eq:Al:SNR}
\frac{\sigma^2_{\mathrm{n}}}{P_{\mathrm{r}}}\gamma_{k}^{\text{(Al)}}=\frac{\Gr(\phir,\phrV,\psrV,k)}{\Lr}\sum_{m=0}^{\Lt-1} \frac{|\gt_m(\phit)|^2}{\Lt},
\end{align}
where $\Gr(\cdot)$ is given by~\eqref{def:Gr}. The average signal power is given in this case by $P_{\mathrm{r}}=\mathbb{E}\{|a(t)\bar{x}_0(t)|^2\}=\mathbb{E}\{|a(t)\bar{x}_1(t)|^2\}$=$\sum_{i=0}^{N_{\textrm{sym}}-1}\mathbb{E}\{|\boldsymbol{D}_{\textrm{a}}\boldsymbol{s}_{i}|^2\}$. 

We can observe from~\eqref{eq:Al:SNR} that optimizing the sum-\gls{SNR} of the Alamouti scheme is equivalent to optimizing $\sum_{k=0}^{K-1}\Gr(\cdot)$ which is the sum-\gls{SNR} of an \gls{ACN} with $\Lr$ antennas. By Theorem~\ref{theo:main} we can conclude that if $\Lr\leq K$ and phase slopes are chosen to satisfy~\eqref{eq:th:optCon:x}, we can achieve
\begin{align}
	\frac{\sigma^2_{\mathrm{n}}}{P_{\mathrm{r}}}\sum_{k=0}^{K-1}\gamma_{k}^{\text{(Al)}}=\SsnrA^\star(\phir,\phit)=\SsnrB^\star(\phir,\phit).
\end{align}
Hence, in this scenario where the channel is modeled as a slowly varying dominant path, implying that the only available spatial diversity of the channel is due to the different far field functions of antennas (the propagation environment has spatial
diversity of order one), Alamouti does not achieve enhanced performance compared to the low-cost, low-complexity \gls{ABN} and \gls{ASN} schemes.

\section{\gls{MRC} Enhancement}
The overall performance for \glspl{VU} can be improved by employing an \gls{MRC} digital combining stage after the analog combining at the receiver. In other words, we use a hybrid combiner at the receiver. We follow a sub-connected configuration as in~\cite{HC} such that antennas are divided to subgroups of $\Lrp$ elements that are combined in analog domain using an \gls{ACN} then fed to a digital port. We would like to show that the solutions provided by Theorem~\ref{theo:main} and Theorem~\ref{th:SA} are still optimal. 
First, after the \gls{ACN}, the signal at port $p$ of a given \gls{VU} can be modeled following~\eqref{eq:r:generic} as
\begin{align}
r_p(t)=&a(t) x(t) c_p(kT)+n_p(t),
\end{align}
where $kT\leq t\leq kT+T_{\textrm{m}}$, and the approximation $c_p(t)\approx c_p(kT)$ follows from the assumptions in \secR~\ref{Sec:sub:assumptions}.
Employing \gls{MRC} with coefficients $c^{*}_p(kT)/\Lrp$~\cite{MRC2002}, and since the noise is uncorrelated between the ports, the overall \gls{SNR} per packet can be expressed as
\begin{align}
\gamma^{\mathrm{(d)}}_k= \sum_{p=0}^{P-1}\frac{P_{\mathrm{r}}}{\Lrp\sigma^2_{\mathrm{n}}}|c_p(kT)|^2.
\end{align}

Given that we are using an \gls{ABN}, the effective channel gain for the $k^{\text{th}}$ packet is modeled by~\eqref{eq:ca}, where $\gr$, $\psr$, $\phr$ are sub-indexed with $(l,p)$ instead of $l$, and $\Lr$ is substituted by $\Lrp$. This can be expressed explicitly as
\begin{align}
c_p^{(\mathrm{a})}(kT)&=\frac{1}{\sqrt{\Lt}} \sum_{m=0}^{\Lt-1}|\gt_m(\phit)|\mathrm{e}^{-\jmath(\pst_m-\pht_mkT)}\nonumber \\
&\times\sum_{l=0}^{\Lrp-1} |\gr_{l,p}(\phir)|\mathrm{e}^{-\jmath(\psr_{l,p}-\phr_{l,p}kT)}, 
\end{align}
where $\gr_{l,p}$, $\phr_{l,p}$, $\psr_{l,p}$ are, respectively, the far field function, phase slope and effective channel phase associated with $l^{\text{th}}$ receive antenna element connected to port $p$. 
We let $\psrV_p$ and $\phrV_p$ be vectors with $\Lrp$ elements corresponding to $\psr_{l,p}$ and $\phr_{l,p}$, respectively.
Then, the normalized sum-\gls{SNR} can be expressed as
\begin{align}
\SsnrD&(\phir,\phit,\phrV,\phtV,\psrV,\pstV)=\sigma^2_{\mathrm{n}}/P_{\mathrm{r}}\sum_{k=0}^{K-1}\gamma^{\mathrm{(d)}}_k\nonumber \\
& =\sum_{p=0}^{P-1} \sum_{k=0}^{K-1}\frac{G_{\text{r},p}(\phir, \phrVp,\psrVp,k)}{\Lrp} \frac{\Gt (\phit, \phtV,\pstV,k)}{\Lt} \nonumber\\
&=\sum_{p=0}^{P-1} \SsnrAp(\phir,\phit,\phrVp,\phtV,\psrVp,\pstV),\label{def:SsnrD}
\end{align}
where $\phrV=[\phrV_0, \phrV_1, \cdots, \phrV_{P-1}]$, $\psrV=[\psrV_0, \psrV_1, \cdots, \psrV_{P-1}]$ and $G_{\text{r},p}(\cdot)$ is defined following~\eqref{def:Gr}, with $\gr_l$ and $\Lr$ replaced by $\gr_{l,p}$ and $\Lrp$, respectively.
The term $\SsnrAp(\cdot)$ is the same as~\eqref{def:S:GplusJ},
\begin{align}
    \SsnrAp(\phir,\phit,\phrVp,\phtV,&\psrVp,\pstV)=K \overline{G}_p(\phir,\phit)\nonumber \\ &+\JAp(\phir,\phit,\phrVp,\phtV,\psrVp,\pstV),\label{def:SsnrAp}
\end{align}
where $\overline{G}_p$ and $\JAp$ are given by~\eqref{def:G} and~\eqref{def:J:original}, respectively, with $\Lr$ replaced by $\Lrp$, and $\gr_l$ by $\gr_{l,p}$ (a sub-index $p$ is added to $c_{l,i},c_{l,i}',d_{m,j},d_{m,j}'$ in~\eqref{def:J:cli:dmj}).

As done earlier in \secR~\ref{sec:design}, we assume that all receiving \glspl{VU} have the same number of antennas, the same far field functions, and additionally, we assume that they have same number of ports. That allow us to define the worst user in the system based on the worst-case \gls{AOA}, \gls{AOD}. Following that,
Theorem~\ref{theo:main} already give us the solutions for any $(\phir,\phit)$ to the per-port subproblems
\begin{align}
(\phrSV_p,\phtSV)=\arg \sup_{(\phrVp,\phtV)} \inf_{(\psrVp,\pstV )} \SsnrAp(\phir,\phit,\phrVp,\phtV,\psrVp,\pstV ),
\end{align}
where,  $\phrVp \in \mathbb{R}^{\Lrp}$, $\phtV \in \mathbb{R}^{\Lt}$, $
\psrVp\in [0,2\pi)^{  \Lrp }$, $\pstV \in [0,2\pi)^{\Lt}$.
The optimal sum-\gls{SNR} per sub-group of antennas is given by 
\begin{align}
\SsnrAp^{\star}(\phir,\phit)&=K \overline{G}_p(\phir,\phit) \nonumber\\
&=K\sum_{l=0}^{\Lrp-1}\frac{|\gr_{l,p}(\phir)|^2}{\Lrp} \sum_{m=0}^{\Lt-1}\frac{|\gt_m(\phit)|^2}{\Lt}.\label{def:Gp}
\end{align}
To deduce the optimum of the overall problem we can make use of the following lemma.
\begin{lemma}\label{lem:Multiport:bound}
	Let $\SsnrD$ be given by~\eqref{def:SsnrD}, then 
	\begin{align}
	\sup_{(\phrV,\phtV)} \inf_{(\psrV,\pstV )}& \SsnrD(\phir,\phit,\phrV,\phtV,\psrV,\pstV )\leq 
	K \sum_{p=0}^{P-1}\overline{G}_p(\phir,\phit).\label{eq:UB:Mp}
	\end{align}
	where $\overline{G}_p(\phir,\phit)$ is given by~\eqref{def:Gp}.
\end{lemma}
\begin{proof}
	See Appendix~\ref{app:multiport}.
\end{proof}
The right-hand side of~\eqref{eq:UB:Mp} is equal to $\sum_{p=0}^{P-1} \SsnrAp^{\star}(\phir,\phit)$, which means that the solutions to the subproblems do satisfy the bound, and therefore are optimal for the overall problem.  
For that to hold, the optimality conditions $\Lt \Lrp\leq K$ and~\eqref{eq:th:optCon} have to be satisfied for all ports $p$. Given that we design the phase slopes according to a certain construction, e.g.,~\eqref{eq:th:ph_slopes}, then we can start with choosing $\phtSV$ and $\phrSV_p$ corresponding to the largest $\Lrp$, then the remaining $\phrSV_{p'}$, $p'\neq p$ can be cloned on the already designed phase slope vector, and all optimality conditions will be satisfied. 

Similarly to what has been found earlier, $\phtSV$ is independent of the far field function of the antennas, and it can be designed according to the maximum number of antennas a \gls{VU} can have $\Lrm$. 
For users with more than one digital port, the condition on the maximum number of antennas is given by $ \Lrp\leq \Lrm$, $\forall p$.
Therefore, a transmitter can use the same \gls{ABN} beamforming vector to achieve an optimal performance for any user (in spite of the number of antennas, ports, and far field functions employed), including the worst one.

Due to limitation of space we will not go through detailed analysis of the \gls{ASN}. However, we point out that since the optimization parameters are restricted to $\phrVp$ and $\psrVp$ ---which are independent from one port to another---by solving the subproblems per port using Theorem~\ref{th:SA}, we already solve the overall problem and achieve the same optimal performance given by
\begin{align}
\SsnrD^{\star}(\phir,\phit)=\sum_{p=0}^{P-1}S_{\text{d},p}^{\star}(\phir,\phit)=K \sum_{p=0}^{P-1}\overline{G}_p(\phir,\phit).
\end{align}
Note that for a system with $\Lrp=\Lr/P$, $\forall p$ we get, 
\begin{align}\label{eq:Sd:star}
\SsnrD^{\star}(\phir,\phit)=PK \overline{G}(\phir,\phit).
\end{align}
Hence, the digital processing stage yields $10\log_{10} (P)~$dB higher gain for any $(\phir,\phit)$ compared to $\SsnrA^\star(\phir,\phit)$. 

\section{Numerical Results}
In this section, we visualise the performance of \gls{ABN}/\gls{ASN} based on some examples of antenna radiation patterns that are shown in \fig~\ref{fig:antennas}. There are two sector antennas $A2$. These are back-to-back patch antennas designed by Smarteq\footnote{Smarteq Wireless AB is a Swedish industrial partner specialized in developing antenna solutions for vehicle industry among others.} for vehicular applications. Moreover, there is an ideal $A0$ and a non-ideal, synthetic, $A1$ omnidirectional antennas. All antenna types have the same average power gain in azimuth plan.

Since the optimal phase slopes of \gls{ABN} were found to be the same for any \gls{VU}, including the worst one, and since \gls{ASN} strategy treats all users the same, we quantify the performance of the two schemes in this section, assuming one transmitting and one receiving, reference \glspl{VU}.
Performance is assessed according to the sum-\gls{SNR} of a burst of $K$ consecutive \gls{CAM} packets, with focus on the worst-case \gls{AOA}, \gls{AOD}, such that we can characterize the robustness of \gls{V2V} cooperative communication against unfavorable angles.
	\begin{figure}[]
	\centering
	\includegraphics[width=0.9\columnwidth]{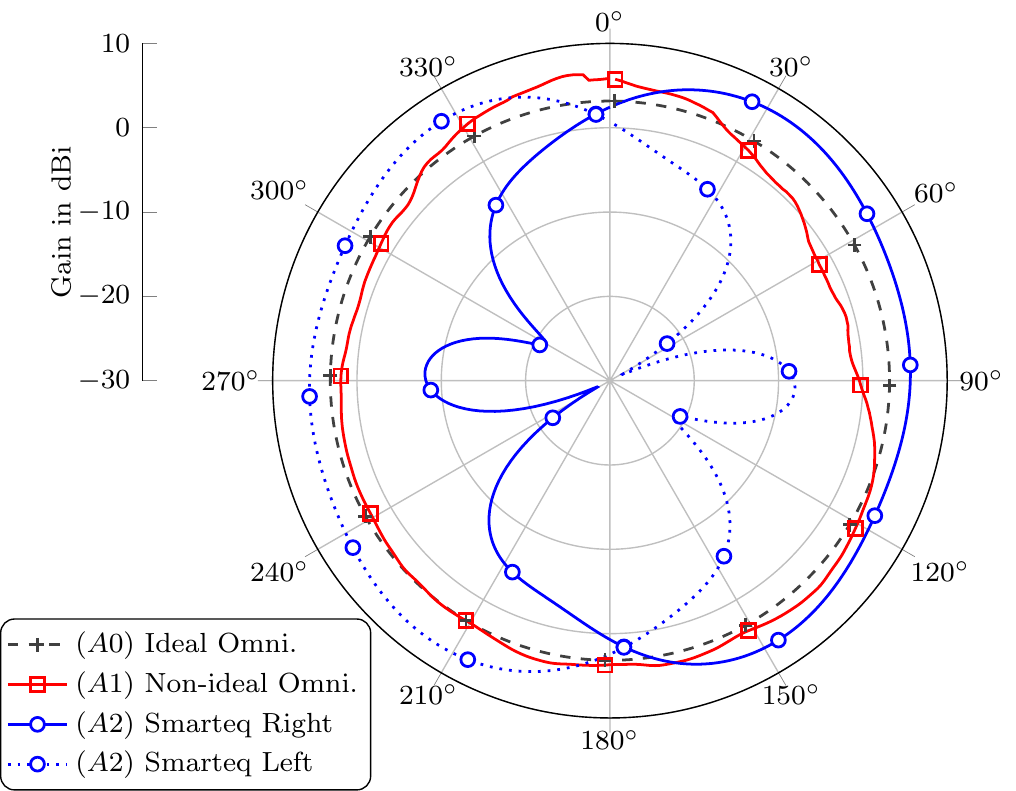}
	\caption{Test antennas radiation patterns, $|g(\phi)|^2$, in dBi. All antennas have the same average power gain in azimuth plan. (Omni. i.e., Omnidirectional).}
	\label{fig:antennas}
\end{figure}
\subsection{Sum-\gls{SNR} Achieved Using \gls{ABN}/\gls{ASN}}
	Consider the use of $A2$ sector antennas (Fig~\ref{fig:antennas}). Employing two $A2$ modules pointing $180~\deg$ apart as shown in the figure, we can enable full omnidirectional coverage. An \gls{ABN}/\gls{ASN} allow us to do that using a single transmit digital port. 
	Note that, with a single digital port, and in absence of \gls{ABN}/\gls{ASN}, the use of one $A2$ module would yield very low performance at the direction where the antenna has low gain.
	Now, let us consider a communication link with $\Lt=2$, $\Lr=2$, and a burst of $K=4\geq \Lr\Lt$ consecutive \gls{CAM} packets. 
	We assume that both \gls{Tx} and \gls{Rx} are equipped with $A2$ antennas. 
	 We plot in \fig~\ref{Fig:comparison} the CDF of the sum-\gls{SNR} of $2\times2$ \gls{ABN}/\gls{ASN}-\gls{ACN} system (with optimal phase slopes) for uniform \gls{AOA}, \gls{AOD}. 
	 We note that the sum-\gls{SNR} is normalized with respect to $P_{\mathrm{r}}/\sigma^2_{\mathrm{n}}$, following~\eqref{def:S:GplusJ} and~\eqref{eq:Sb:general:definition}, and that holds throughout this section.
	 As indicated in Remark~\ref{remark:theorem:main} the sum-\gls{SNR} CDF is the same for any $\psrV \in [0,2\pi)^{\Lr}, \pstV \in [0,2\pi)^{\Lt} $, and it is identical for both \gls{ABN} and \gls{ASN} which coincides with our analytical result, $\SsnrA^\star(\phir,\phit)=\SsnrB^\star(\phir,\phit)$. 
	 We recall that the optimal sum-\gls{SNR} is proportional to the equivalent radiation pattern $\overline{G}(\phir,\phit)$, defined in~\eqref{def:G}. 
	 To observe how this is characterized at the transmitter side we visualize in \fig~\ref{fig:Geq} the term $\sum_{m=0}^{\Lt-1}|\gt_m(\phit)|^2/\Lt$, which can be seen as the transmitter side equivalent pattern. We see that it has a nearly ideal omnidirectional coverage. 
	 The receiver equivalent pattern $\sum_{l=0}^{\Lr-1}|\gr_l(\phir)|^2/\Lr$, which is omitted from the figure, has the same characteristics. 
	 From this example we see that \gls{ABN}/\gls{ASN} yields an equivalent radiation pattern with improved omnidirectional characteristics (that implies improved robustness against unfavorable \gls{AOA}, \gls{AOD}). 
	 
     We mentioned earlier that due to distortions caused by many factors, including vehicle body, a practical omnidirectional pattern is far from ideal. 
	 As an example of how a non-ideal omnidirectional antenna performs in comparison to an ideal one, we see in \fig~\ref{Fig:comparison} that $1\times1$ $A1$ (at both \gls{Tx} and \gls{Rx}) system has much lower worst-case performance than $1\times 1$ $A0$, despite that both antenna types have the same average power gain in azimuth plane.
	 The \gls{ABN}/\gls{ASN} can be used to improve the performance of a system based on such non-ideal antennas as well. In particular, an improvement of around $4.5~$dB in worst-case sum-\gls{SNR} can be achieved using $2\times 2~ A1 $ \gls{ABN}/\gls{ASN} system compared to $1\times 1~A1$ system. This entails that \gls{ABN}/\gls{ASN} results in improved robustness of the system, which is also evident from the improved omnidirectional characteristics of the resultant equivalent radiation pattern at the \gls{Tx} side, shown in \fig~\ref{fig:Geq}.
	 \begin{figure}[]
 	\centering
 		\includegraphics[width=\columnwidth]{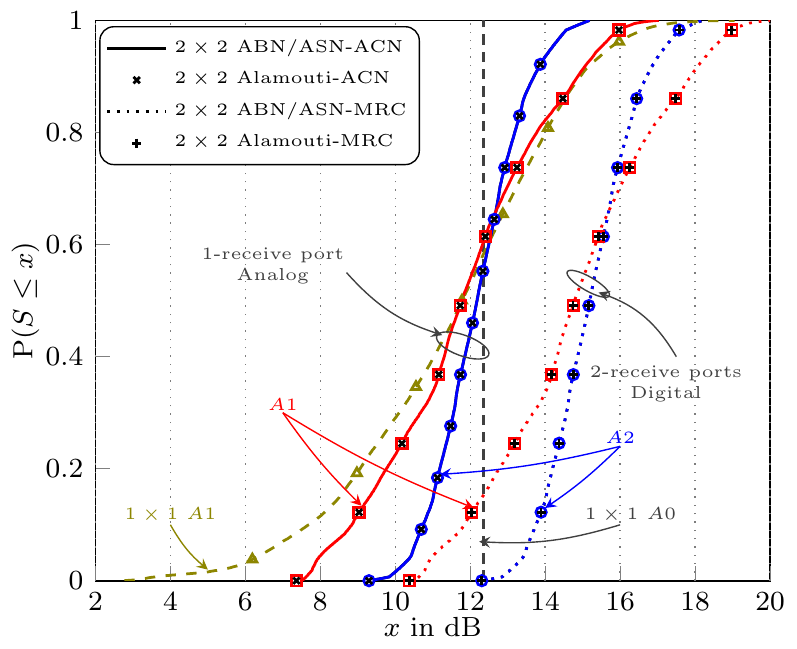}
 	\caption{CDF corresponding to the normalized sum-\gls{SNR} of $K=4$ packets, for different antennas and signal processing choices. The CDFs are computed for uniform \gls{AOA}, \gls{AOD}.}
 	\label{Fig:comparison}
 \end{figure}
	\begin{figure}[]
	\centering
	\includegraphics[width=0.85\columnwidth]{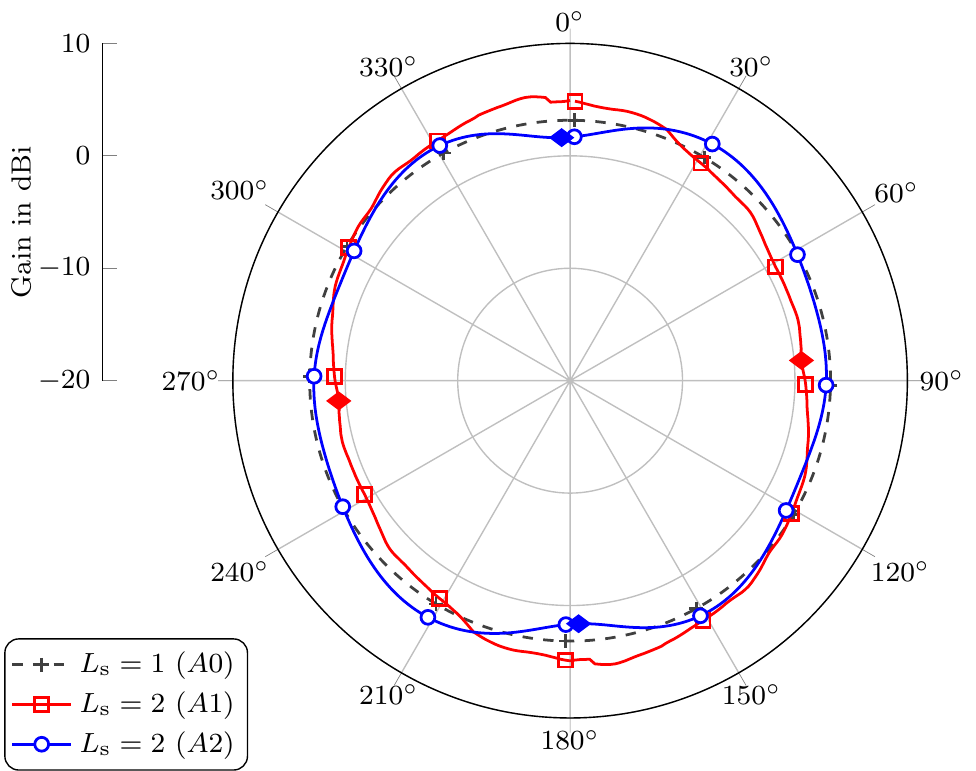}
	\caption{Equivalent radiation pattern realized at the transmitter side using \gls{ABN}/\gls{ASN} for $2$ elements of either $A1$ or $A2$ antennas. The elements are pointing $180~\deg$ apart (i.e., $\gt_0(\phit)=\gt_1(\phit-180)$). Worst-case \glspl{AOD} are marked with diamonds.}
	\label{fig:Geq}
\end{figure}  
		
Given access to more than one digital receive port and \gls{MRC} processing, the \gls{ABN}/\gls{ASN} sum-\gls{SNR} can be enhanced by $3~$dB for all \gls{AOA}, \gls{AOD}. This can be seen from the $3~$dB shifted CDF of \gls{ABN}/\gls{ASN}-\gls{MRC} shown in \fig~\ref{Fig:comparison} compared to \gls{ABN}/\gls{ASN}-\gls{ACN} CDF. This $3~$dB gain is in accordance with~\eqref{eq:Sd:star} when setting $\Lrp=1,~ p=0,1$. From another aspect, we recall that an \gls{ABN}/\gls{ASN} has the same performance as an Alamouti scheme. Therefore, given access to two digital transmit ports, and as shown by the CDFs of Alamouti, no enhanced performance is achieved compared to \gls{ABN}/\gls{ASN} in this case.
	 \subsection{On Optimality of Phase Slopes}
	 \begin{figure}[]
	\centering
		\includegraphics[width=0.95\columnwidth]{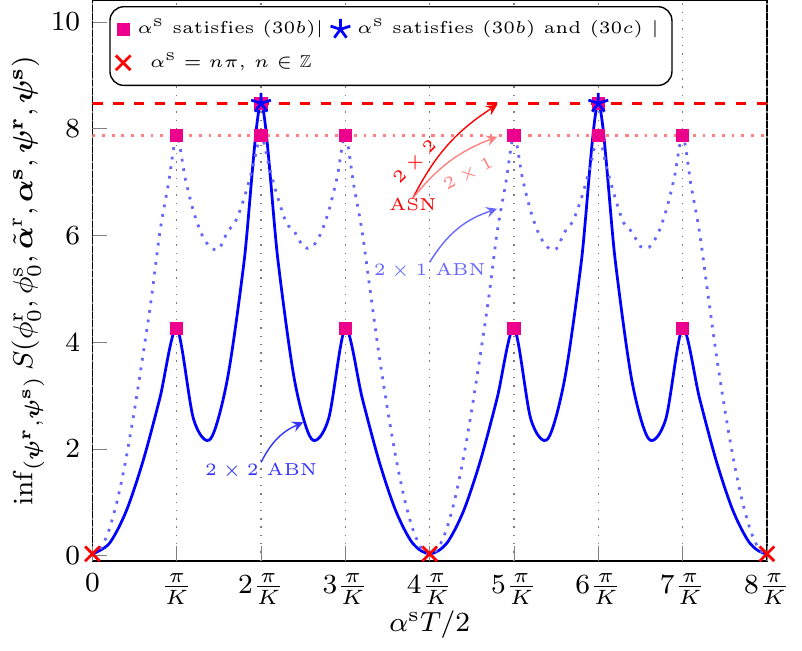}
			\caption{Normalized sum-\gls{SNR} of $K=4$ packets of \gls{ABN} and \gls{ASN} as a function of $\phtV=[0, \pht]^{\textsf{T}}$ when $\Lt=2$ $(A2)$, and $\Lr=2$ $(A2)$ with $\phrV=[0,2\pi/KT]^{\textsf{T}}$, or $\Lr=1$ $(A2)$. The \gls{AOA}, \gls{AOD} are fixed to $(\phir_0,\phit_0)=(178,178)~\deg$.} 
		\label{fig:conditions}
\end{figure}
	 
	 In Theorem~\ref{theo:main} we derived the optimality condition~\eqref{eq:th:optCon}. To get insight into it,
	 we plot in \fig~\ref{fig:conditions} sum-\gls{SNR} of a $2\times 2$ \gls{ABN} system based on $A2$ antennas as a function of transmitter phase slope for a fixed \gls{AOA}, \gls{AOD} (recall that Theorem~\ref{theo:main} holds for any $(\phir, \phit)$). An \gls{ACN} with phase slopes vector that satisfy~\eqref{eq:th:optCon:x} is used. In particular, the \gls{ACN} vector is chosen according to~\eqref{eq:th:ph_slopes},~\eqref{eq:th:SA:ph_slopes}, so it is optimal for \gls{ASN} as well. In the same figure we plot also the sum-\gls{SNR} for the \gls{ASN} as a reference. Since \gls{ASN} does not depend on transmitter phase slopes it maintains optimal performance. As for the $2\times2$ \gls{ABN} we can observe that when only~\eqref{eq:th:optCon:z} is satisfied, a suboptimal performance is achieved. Optimality is ensured, however, when both~\eqref{eq:th:optCon:z} and~\eqref{eq:th:optCon:xz} are met, which is in accordance with the results of Theorem~\ref{theo:main}. 
	 In comparison to $2\times1$ \gls{ABN}, 
	 we see that at the points where~\eqref{eq:th:optCon:z} is satisfied, optimal performance is achieved.
	  Thus, employing multiple antennas at both \gls{Tx} and \gls{Rx} sides reduces the set of optimal phase slopes points. 
	 Consequently, we see in the figure that for a large deviation from the optimal phase slopes, the relative sum-\gls{SNR} loss is larger for the $2\times2$ than the $2\times 1$ system. Yet, for a small phase slope mismatch, the relative loss is small for both systems.
	Curves showing the effect of \gls{ACN} phase slopes mismatch on the sum-\gls{SNR} of both \gls{ABN} and \gls{ASN} follow similar trends, and they are omitted for clarity.
	 Since the \gls{ASN}-\gls{ACN} structure depends only on phase slopes at the receiver, then it is less susceptible to effects of mismatch in phase slopes compared to \gls{ABN}-\gls{ACN}.

\begin{figure}[]
	\centering
	\includegraphics[width=0.95\columnwidth]{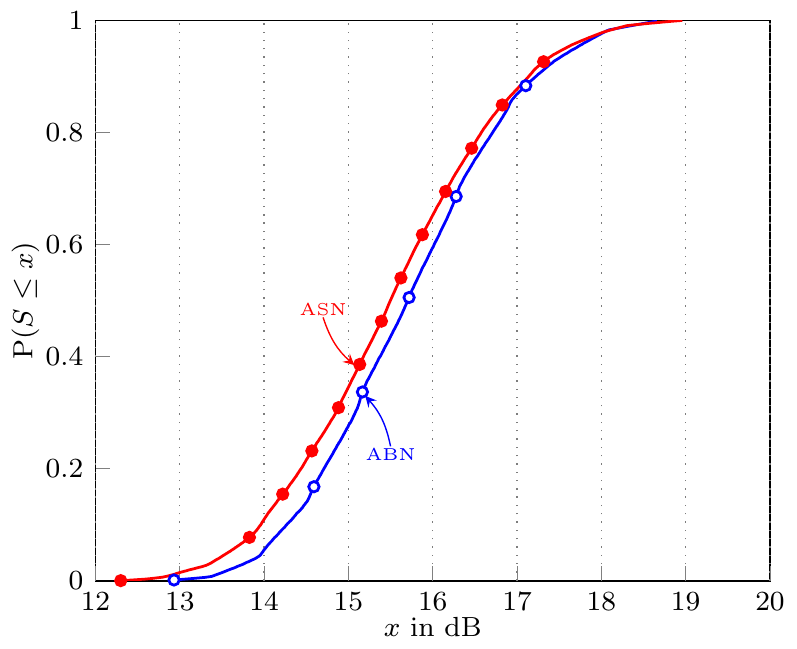}
	\caption{CDF corresponding to the normalized sum-\gls{SNR} of $K=9$ packets of \gls{ABN} and \gls{ASN}~\eqref{eq:Sstar:vir:counter}, $\Lr=2$, $\Lt=2$. The CDFs are computed for uniform \gls{AOA}, \gls{AOD}.} 
		\label{Fig:ASNvir}
\end{figure}

\subsection{\gls{ASN} Performance when $K/\Lt\notin \mathbb{Z}$}
As stated in \secR~\ref{section:sub:ASN:MA}, in the case $K/\Lt\notin \mathbb{Z}$, \gls{ASN} is not known to achieve optimal performance achieved by \gls{ABN}, when \gls{ACN} phase slopes are set according to~\eqref{eq:th:SA:ph_slopes} (the phase slopes are optimal for \gls{ASN} when $\Lt/K\in \mathbb{Z}$ and \gls{ABN}). In such case, \gls{ASN} performance is governed by~\eqref{eq:Sstar:vir:counter}, to which we do not have an available analytical solution. 
To quantify this numerically, we consider a burst of $K=9$ packets, and a communication link between a transmitting \gls{VU} with $\Lt=2$ and a reference receiving \gls{VU} with $\Lr=2$. Both \gls{Tx} and \gls{Rx} uses $A2$ antennas. In \fig~\ref{Fig:ASNvir} we show the CDFs of the achieved sum-\gls{SNR}, when the transmitting \gls{VU} is employing \gls{ASN} or \gls{ABN}, while the receiving \gls{VU} is employing an \gls{ACN} (following~\eqref{eq:th:SA:ph_slopes}). 
We observe that \gls{ASN} achieves a suboptimal performance when $K/\Lt\notin\mathbb{Z}$. Yet, for this particular antenna system, the performance is negligibly lower (by around $ 0.5~$dB) than the optimal \gls{ABN} worst-case sum-\gls{SNR}. 

	\section{Conclusions}
		A fully analog, low-complexity, multiple antenna system for cooperative, periodic, broadcast \gls{V2V} communication has been presented. Given that vehicular users (\glspl{VU}) use an \gls{ACN}~\cite{ACN} with $\Lr$ antennas at the receiver, we proposed the use of a network of phase shifters, \gls{ABN}, or a network of switches, \gls{ASN}, with $\Lt$ antennas at the transmitter. The overall system has been optimized, in absence of channel knowledge, to maximize the sum-\gls{SNR} of $K$ consecutive packets for the worst receiving \gls{VU} (which minimizes the burst error probability when the packet error probability decreases exponentially with the per-packet SNR).
		The main findings of this work follow.
		\begin{itemize}
		
		\item  An optimal set of phase slopes for \gls{ABN}~\eqref{eq:th:optCon}, and \gls{ASN}~\eqref{eq:th:SA:OptCon} are derived when $\Lr\Lt\leq K$. 
		The phase slopes are found to be independent of the far field functions of the antennas. Moreover, the receiver phase slopes are consistent with what was found in~\cite{ACN}.
		
		\item Both structures yield the same optimal sum-\gls{SNR} in the general case. For the special case of $\Lt=2$ transmit antennas, the sum-\gls{SNR} is equivalent to what can be achieved using an Alamouti diversity scheme. 
			
			
		\item Given the maximum number of receive antennas a \gls{VU} can have, $\Lrm$, the \gls{ABN} can be designed to achieve optimal performance for any user with $\Lr\leq \Lrm$, including the worst receiving user. 
		The \gls{ASN} on the other hand, does not require knowledge of $\Lrm$, and it achieves optimal performance as long as $\Lt$ divides $K$. 
					
			 
		\item An \gls{ACN} receiver can be designed to communicate optimally with both \gls{ABN} and \gls{ASN} transmitters.
			
		\item The derived phase slopes for \gls{ABN}/\gls{ASN} structures guarantee also optimal sum-\gls{SNR} for a hybrid analog-digital \gls{ACN}-\gls{MRC} combiner. 

		
		\end{itemize}


\begin{appendices}
	\section{Proof of Theorem~\ref{theo:main}}\label{app:theo:main}
\subsection{Notation}
For convenience, we introduce the following notation. 
For an $N$-vector
$\boldsymbol{v} = [ v_0 , v_1 ,\cdots , v_{N-1} ]^\textsf{T} $we define the $(N-1)$-vector $\bar{\boldsymbol{v}}_{n}$ as
\begin{equation}
\label{eq:v:bar:n:def}
\bar{\boldsymbol{v}}_{n} \triangleq 
[v_0 , v_1 ,\cdots , v_{n-1} ,~ v_{n+1} , \cdots , v_{N-1}]^\textsf{T}.
\end{equation}
We define the ``averaging'' operator $\Av$ as an integral of a function with respect to variables $\boldsymbol{y}\in[0,2\pi)^{\Lr}$ and $\boldsymbol{v}\in[0,2\pi)^{\Lt}$, such that
\begin{equation}
\label{eq:A:def}
\Av[h(\yrV, \vtV)] \triangleq \int_{[0, 2\pi)^{\Lt}}\int_{[0, 2\pi)^{\Lr}} h(\yrV, \vtV) \, d\yrV d\vtV.
\end{equation}
If we restrict the integration to be over just $\boldsymbol{y}$ or $\boldsymbol{v}$, we use the notation $\Av_{\boldsymbol{y}}$ and $\Av_{\boldsymbol{v}}$, respectively. We note that 
\begin{align}\label{def:Av:decomposition}
\Av[\cdot]=\Av_{\boldsymbol{v}}\bigl[\Av_{\boldsymbol{y}}[\cdot]\bigr]= \Av_{\boldsymbol{y}}\bigl[\Av_{\boldsymbol{v}}[\cdot]\bigr].
\end{align}
Finally, we define
\begin{align}
\PsetR &\triangleq \{(l, i): l, i\in\mathbb{Z}, 0\le l < i \le \Lr-1 \}, \label{eq:Pr:def}\\
\PsetT &\triangleq \{(m, j): m, j\in\mathbb{Z}, 0\le m < j \le \Lt-1 \}, \label{eq:Ps:def}\\
\Delta y_{l, i}  &\triangleq y_i - y_l. \label{eq:Delta:def}
\end{align}
Note that $\PsetR=\emptyset$ when $\Lr=1$, and $\PsetT=\emptyset$ when $\Lt=1$. In accordance with this, we use the convention
	\begin{align}\label{conv:ZeroSet:sum}
	\sum_{e\in \mathcal{S}}c_e=0, \quad\textrm{ when }~\mathcal{S}=\emptyset.
	\end{align}
	
\subsection{Preliminaries}\label{app:f:properties}
For a fixed constant positive integer $K$ we define the function $f: \mathbb{R}^2\rightarrow \mathbb{R}$ as
\begin{align}\label{eq:f:def}   
f(x,y) \triangleq \sum_{k=0}^{K-1} \cos(y-2kx).
\end{align}
Noting that $f(x,y)=\text{Re}\{ \mathrm{e}^{\jmath y}\sum_{k=0}^{K-1} \mathrm{e}^{-\jmath 2kx}\}$ and using the sum of geometric series,
we get
\begin{align}\label{eq:f:cases:app2}
f(x,y)=
\begin{cases}
\displaystyle K \cos(y),& x\in\mathcal{X}\\
\displaystyle 	\frac{\sin(Kx)}{\sin(x)}\cos\big(y-(K-1)x\big),& x\notin\mathcal{X}
\end{cases}
\end{align}
where $\mathcal{X}\triangleq\{q\pi:q\in \mathbb{Z}\}$.

We state the following properties of $f(x,y)$:
\begin{enumerate}
	\item $f(x, y)=0$ for all $y\in\mathbb{R}$, if and only if $x\in\mathcal{X}^*$, where
	\begin{equation}
	\mathcal{X}^* \triangleq \{q\pi/K: q\in\mathbb{Z}\}\setminus \mathcal{X}, \label{eq:f:x:zeros:prop}
	\end{equation}
	which follows from \eqref{eq:f:cases:app2}. Furthermore, we note that $x\in\mathcal{X}^*\Rightarrow x\notin\mathcal{X}$.
	\item from the identity $2\cos(a)\cos(b)=\cos(a+b)+\cos(a-b)$ and \eqref{eq:f:def}, we deduce
	\begin{align}
	2&\sum_{k=0}^{K-1} \cos(y_1-2kx_1)\cos(y_2-2kx_2)
	= \nonumber\\
	&f(x_1+x_2, y_1+y_2) + f(x_1-x_2, y_1-y_2). \label{eq:f:add:prop}
	\end{align}
	\item suppose $a$ is a nonzero integer, then for all $x, B, C\in\mathbb{R}$, 
	\begin{equation}
	\int_{0}^{2\pi} B f(x, ay+C)\, dy 
	= 0, \label{eq:f:zero-mean:prop}
	\end{equation}
	which follows from \eqref{eq:f:def}.
\end{enumerate}

	\begin{lemma}
	\label{lem:f:delta:zeromean}
	Let $\yrV = [\yr_0,\yr_1,\cdots,\yr_{\Lr-1}]^\textsf{T}\in[0,2\pi)^{\Lr}$ 
	and 
	$\displaystyle \vtV = [ \vt_0, \vt_1,\cdots,
	\vt_{\Lt-1}]^\textsf{T}\in[0,2\pi)^{\Lt}$. Then for $l,i\in\{0,
	1, \cdots, \Lr-1\}$ and $m,j\in\{0, 1, \cdots, \Lt-1\}$
	\begin{align}
	\Av_{\yrV}[f(x, \Delta \yr_{l, i}+C)] &= 0, \quad l\neq i\label{eq:f:y:zero-mean}\\
	\Av_{\vtV}[f(x, \Delta \vt_{m, j}+C)] &= 0, \quad m\neq j\label{eq:f:v:zero-mean}\\
	\Av[f(x, \Delta \yr_{l, i} + \Delta \vt_{m, j}+C)] &= 0, ~ \begin{cases}
	l\neq i \text{, or}\\
	m\neq j
	\end{cases}\label{eq:f:yv:zero-mean}
	\end{align}
	where $x, C\in \mathbb{R}$.
\end{lemma}

\begin{proof}
	Suppose $l\neq i $ we can write
	$f(x, \Delta \yr_{l, i} + \Delta \vt_{m, j}+C) = f(x,\yr_i+ C')$, where $C'=-\yr_l+(\vt_j-\vt_m)+C$ is a constant with respect to $\yr_i$. Then, $\Av_{\yrV}[f(x, \Delta \yr_{l, i} + \Delta \vt_{m, j}+C)]$ can be expressed as
	\begin{align}\nonumber
	\int_{[0, 2\pi)^{\Lr-1}}\underbrace{\left[ \int_{\yr_i \in [0, 2\pi)} f(x, \yr_i +  C' )\, d\yr_i \right]}_{\text{$=0$ according to
			\eqref{eq:f:zero-mean:prop} }}\, d\bar{\yrV}_i
	= 0, 
	\end{align}
	which, together with~\eqref{def:Av:decomposition}, implies that~\eqref{eq:f:yv:zero-mean} holds when $l\neq i$. 
	It is straightforward to develop a similar argument based on $\Av_{\vtV}[\cdot]$ to prove~\eqref{eq:f:yv:zero-mean} when $m\neq j$.
	Moreover, \eqref{eq:f:y:zero-mean} and \eqref{eq:f:v:zero-mean} follows from~\eqref{eq:f:yv:zero-mean} by considering
	$m=j$ and $l=i$, respectively, which completes the proof.
\end{proof}

\begin{lemma}
	\label{lem:f:delta:f:delta:zeromean}
	Let $\yrV = \displaystyle[\yr_0,\yr_1,\cdots,\yr_{\Lr-1}]^\textsf{T}\in[0,2\pi)^{\Lr}$ and $\displaystyle \vtV = [ \vt_0, \vt_1,\cdots,
	\vt_{\Lt-1}]^\textsf{T}\in[0,2\pi)^{\Lt}$.
	For $w,w'\in\PsetR$ and $\pPt,\pPt'\in\PsetT$, we have
	\begin{align}
	\Av_{\yrV}[f(x_1,  C_1\pm\Delta \yr_{w} )f(x_2, C_2\pm\Delta \yr_{w'} )]  =& 0, ~w\neq w' \label{eq:fy:fy:zero-mean}\\
	\Av_{\vtV}[f(x_1,  C_1\pm\Delta \vt_{\pPt})f(x_2, C_2\pm\Delta \vt_{\pPt'} )]  =& 0, ~\pPt\neq \pPt'\label{eq:fv:fv:zero-mean}\\
	\Av[f(x_1, \Delta \yr_{w} + a\Delta \vt_{\pPt})f(x_2, \Delta \yr_{w'}+ a\Delta \vt_{\pPt'}&)]  = 0,\nonumber\\
	\qquad \quad \qquad  a=\pm 1,~w\neq w' &\textrm { or } \pPt\neq \pPt'\label{eq:fyv:fyv:zero-mean}\\
\Av[f(x_1, \Delta \yr_{w} + \Delta \vt_{\pPt})f(x_2, \Delta \yr_{w'}- \Delta \vt_{\pPt'}&)]  = 0, \label{eq:fyv:-fyv:zero-mean}\\
	\Av_{\yrV}[f(x_1, \Delta \yr_{w} \pm \Delta \vt_{\pPt})f(x_2, \Delta \vt_{\pPt'} )]  =& 0, \label{eq:fyv:fv:zero-mean}\\
	\Av_{\vtV}[f(x_1, \Delta \yr_{w} \pm \Delta \vt_{\pPt} )f(x_2, \Delta \yr_{w'})]  =& 0, \label{eq:fyv:fy:zero-mean}\\
		\Av[f(x_1, \Delta \yr_{w} )f(x_2, \Delta \vt_{\pPt} )]  =& 0, \label{eq:fy:fv:zero-mean}
	\end{align}
	where $x_1, x_2, C_1, C_2 \in \mathbb{R}$.
\end{lemma}

\begin{proof}
Let $w=(l,i)$, $w'=(r,q)$, $\pPt=(m, j)$ and $\pPt'= (t, s)$ throughout this proof. Given that $w,w'\in\PsetR$, and $\pPt,\pPt'\in\PsetT$, it follows that $l<i$, $r<q$, $m<j$, and $t<s$.
	
	Let us start by proving~\eqref{eq:fy:fy:zero-mean}.  Suppose $w\neq w'$, i.e., $(l, i)\neq(r, q)$, which implies that $l\neq r$ or $i\neq q$. Let us
	consider these two cases separately.
	
		(i) Suppose that $l\neq r$. Then either $l < r$ or $r < l$, which (since $l<i$ and $r<q$) implies that $n = \min\{l, r\}$ is the unique minimum element in the
		multiset $[l, i, r, q]$. Hence, the integration variable $y_n$ appears in the argument of either
		$f(x_1, C_1\pm\Delta y_{w} )=f(x_1, C_1\pm\Delta y_{l, i} )=f(x_1, C_1\pm(y_i - y_l) )$ or $f(x_2, C_2\pm\Delta y_{w'} )=f(x_2, C_2 \pm\Delta y_{r, q} ) =f(x_2, C_2\pm(y_q - y_r))$, but not in both. 
		Then, the integrand in~\eqref{eq:fy:fy:zero-mean} can be expressed as
		\begin{align}
		    B f(x', C'-ay_n), ~a=\pm1\nonumber
		\end{align}
		where $B=f(x_2,C_2\pm \Delta y_{r, q} )$, $C'=C_1+ay_i$ and $x'=x_1$ if $n=l$, while $B=f(x_1, C_1\pm\Delta y_{l, i} )$, $C'=C_2+ay_q$ and $x'=x_2$ if $n=r$. Note that in both cases $B$ and $C'$ are constants with respect to $y_n$.
		Hence, the left-hand side of~\eqref{eq:fy:fy:zero-mean} satisfies
		\begin{align}
	\int_{[0, 2\pi)^{\Lr-1}}  \underbrace{\left[\int_{y_n\in[0, 2\pi)} B f(x', C'-ay_n) dy_n \right]}_{\text{$=0$ according to
				\eqref{eq:f:zero-mean:prop} }} \, d\bar{\boldsymbol{y}} _n=0, \nonumber 
		\end{align}
		and thus,~\eqref{eq:fy:fy:zero-mean} holds when $l\neq r$.

		(ii) Suppose that $i\neq q$. Then either $i > q$ or $q > i$, which implies that $n = \max\{i, q\}$ is the unique maximum element in
		the multiset $[l, i, r, q]$. Hence, just as in case (i), the integration variable $y_n$ appears in the argument of either $f(x_1,C_1\pm (y_i - y_l ))$ or 
		$f(x_2, C_2 \pm (y_q - y_r) )$, but not in both. Hence, \eqref{eq:fy:fy:zero-mean} holds also when $i\neq q$.
	Combining the results of (i) and (ii) proves that~\eqref{eq:fy:fy:zero-mean} holds. Moreover, the same basic arguments, (i) and (ii), can be repeated
	to prove~\eqref{eq:fv:fv:zero-mean}.
	
	Now, consider the integrand in~\eqref{eq:fyv:fyv:zero-mean}. When integrating over $\boldsymbol{y}$ ($\Av_{\boldsymbol{y}}[\cdot]$), it can be expressed as
	\begin{align}\nonumber
	\quad f(x_1, \Delta y_{w}  + C_1)f(x_2, \Delta y_{w'} + C_2), 
	\end{align}
	where $C_1=a\Delta v_\pPt=a\Delta v_{m,j}=a(v_j - v_m)$ and $C_2=a\Delta v_{\pPt'}=a\Delta v_{t,s}=a(v_s - v_t)$ are constants and $a=\pm 1$. Hence, it follows from~\eqref{eq:fy:fy:zero-mean} and the decomposition $\Av[\cdot]=\Av_{\boldsymbol{v}}\big[ \Av_{\boldsymbol{y}}[\cdot]  \big]$ that~\eqref{eq:fyv:fyv:zero-mean}
	holds when $w\neq w'$. 
	On the other hand, by starting the integration over $\boldsymbol{v}$, then employing~\eqref{eq:fv:fv:zero-mean} and $\Av[\cdot]=\Av_{\boldsymbol{y}}\big[ \Av_{\boldsymbol{v}}[\cdot]  \big]$, we deduce that~\eqref{eq:fyv:fyv:zero-mean}
	holds also when $\pPt\neq \pPt'$, which completes the proof of~\eqref{eq:fyv:fyv:zero-mean}.

		By similar argument to what was presented to demonstrate~\eqref{eq:fyv:fyv:zero-mean}, we can show that~\eqref{eq:fyv:-fyv:zero-mean} holds when $w\neq w'$ or $\pPt\neq \pPt'$. For the case $w=w'$ and $\pPt=\pPt'$, i.e., $(l,i)=(r,q)$ and $(m,j)=(t,s)$, the integrand of~\eqref{eq:fyv:-fyv:zero-mean} can be expressed based on~\eqref{eq:f:def} and the identity $2cos(a)cos(b)=cos(a+b)+cos(a-b)$ as
		\begin{align}
		&\frac{1}{2}\sum_{k=0}^{K-1}\sum_{k'=0}^{K-1}\cos(2\Delta y_{l, i} -2(kx_1+k'x_2)    )\nonumber \\
		&+\cos\big(2\Delta v_{m, j} -2(kx_1-k'x_2)   \big ). \label{eq:fyv:-fyv:integrand}
		\end{align}
		Integrating~\eqref{eq:fyv:-fyv:integrand} with respect to $\boldsymbol{y}$ and $\boldsymbol{v}$ yields zero, and this completes the proof of~\eqref{eq:fyv:-fyv:zero-mean}.

	Since $\Delta v_{\pPt}$ and $f(x_2, \Delta v_{\pPt'})$ are constants with respect to $\boldsymbol{y}$, the left-hand side of~\eqref{eq:fyv:fv:zero-mean} can be written as
	\begin{align}\nonumber 
		f(x_2, \Delta v_{\pPt'} ) \Av_{\boldsymbol{y}}[f(x_1, \Delta y_{w} +C )],
	\end{align}
	where $C=\pm\Delta v_{\pPt}$, and it follows directly from Lemma~\ref{lem:f:delta:zeromean},~\eqref{eq:f:y:zero-mean} ($w=(l,i)$, $l<i$) that~\eqref{eq:fyv:fv:zero-mean} holds. 
	By similar argument~\eqref{eq:fyv:fy:zero-mean} follows from Lemma~\ref{lem:f:delta:zeromean},~\eqref{eq:f:v:zero-mean}. Finally, employing the decomposition~\eqref{def:Av:decomposition},~\eqref{eq:fy:fv:zero-mean} follows from Lemma~\ref{lem:f:delta:zeromean},~\eqref{eq:f:y:zero-mean} or~\eqref{eq:f:v:zero-mean} and this completes the lemma proof.
	\end{proof}

	\subsection{Proof of Theorem 1 and Corollary~\ref{cor:ABN:ph_slopes}}
	
	Given the objective function $\SsnrA(\phir,\phit,\phrV,\phtV,\psrV,\pstV )$, defined in~\eqref{def:S:GplusJ}, where  $\Lr\geq 1$, $\Lt\geq1$, and $\Lr+\Lt>2$ (the trivial case $\Lr=\Lt=1$ is omitted).
	We define $\Lr$-vectors $\boldsymbol{x}\in \mathbb{R}^{\Lr}$ and $\boldsymbol{y} \in [0,2\pi)^{\Lr}$ with elements
	$x_l\triangleq\phr_lT/2$ and $y_l\triangleq\psr_l$, $0\leq l\leq \Lr-1$. Analogously, we define $\Lt$-vectors $\boldsymbol{z}\in \mathbb{R}^{\Lt}$ and $\boldsymbol{v}\in [0,2\pi)^{\Lt}$ with elements $z_m\triangleq\pht_mT/2$ and $v_m\triangleq\pst_m$, $0\le m\le \Lt-1$. Last, we let $\phiVrs \triangleq(\phir,\phit)$. The objective function can be expressed using this notation as
	\begin{align}\label{eq:SsnrA:xy}
	\SsnrA(\phiVrs,\boldsymbol{x},\boldsymbol{z},\boldsymbol{y},\boldsymbol{v} )=K\overline{G}(\phiVrs)+ \JA(\phiVrs,\boldsymbol{x},\boldsymbol{z},\boldsymbol{y},\boldsymbol{v} ),
	\end{align}
	where $ \JA$ is given by~\eqref{def:J:original}. 
	Using the definitions of $f$, $\PsetR$, and $\PsetT$, we can express $\JA$ as 
	\begin{align}
	&\JA(\phiVrs,\boldsymbol{x},\boldsymbol{z},\boldsymbol{y},\boldsymbol{v} )=\nonumber \\
	&\overbrace{\sum_{w\in \PsetR} c_{w} f(\Delta x_{w},\Delta y_{w})}^{J_{\mathrm{a}}^{\mathrm{r}}} + \overbrace{\sum_{\pPt\in \PsetT} d_{\pPt}f(\Delta z_{\pPt},\Delta v_{\pPt})}^{J_{\mathrm{a}}^{\mathrm{s}}}\nonumber\\
	& + \sum_{w\in \PsetR}\sum_{\pPt\in \PsetT} 0.5c'_{w}d'_{\pPt} f(\Delta x_{w}+\Delta z_{\pPt},\Delta y_{w}+\Delta v_{\pPt}) \nonumber\\
	&+ \sum_{w\in \PsetR}\sum_{\pPt\in \PsetT} 0.5c'_{w}d'_{\pPt} f(\Delta x_{w}-\Delta z_{\pPt},\Delta y_{w}-\Delta v_{\pPt}),\label{eq:J:wp:compact}
	\end{align}
	where the last two terms follow from~\eqref{eq:f:add:prop}. The dependency of the nonnegative coefficients $c_w$, $c'_w$, $d_{\pPt}$, $d'_{\pPt}$---defined in~\eqref{def:J:cli:dmj}---on $\phiVrs$ is dropped for convenience. Note that, in the special case of $\Lr=1$, $\PsetR=\emptyset$ and with reference to the convention~\eqref{conv:ZeroSet:sum}, we get $\JA=J_{\mathrm{a}}^{\mathrm{s}}$, while when $\Lt=1$, $\PsetT=\emptyset$, we get $\JA=J_{\mathrm{a}}^{\mathrm{r}}$.
		
		We observe that $\JA(\phiVrs,\boldsymbol{x},\boldsymbol{z},\boldsymbol{y},\boldsymbol{v} )$ is a linear combination of $f$-functions with nonnegative
	coefficients---a crucial property for the proof to follow. 
The proof is divided to three lemmas, which are presented first, then follows the demonstration of the theorem claims at last.
	\begin{lemma}\label{lem:inf}
		Let $\JA$ be defined as in~\eqref{eq:J:wp:compact} where 
		$\boldsymbol{x}\in \mathbb{R}^{\Lr}$,  
		$\boldsymbol{y}\in [0,2\pi)^{\Lr}$,
		$\boldsymbol{z}\in \mathbb{R}^{\Lt}$,
		and $\boldsymbol{v}\in [0,2\pi)^{\Lt}$.
		Then, 
		\begin{align}
		\Av[\JA(\phiVrs,\boldsymbol{x},\boldsymbol{z},\boldsymbol{y},\boldsymbol{v} )]&=0,\label{eq:lem:inf:AJeq:zero}\\
		\inf_{\boldsymbol{y}\in [0,2\pi)^{\Lr},\boldsymbol{v}\in [0,2\pi)^{\Lt}} \JA(\phiVrs,\boldsymbol{x},\boldsymbol{z},\boldsymbol{y},\boldsymbol{v} ) &\leq 0.\label{eq:lem:inf:claim}
		\end{align}
	\end{lemma}
	
	\begin{proof}
		From~\eqref{eq:J:wp:compact} we see that $\JA(\phiVrs,\boldsymbol{x},\boldsymbol{z},\boldsymbol{y},\boldsymbol{v})$ is a linear combination of $f$-functions with nonnegative coefficients. The $f$-functions are of the general form $f(x, a \Delta y_{w} + b \Delta v_{\pPt})=f(x, a \Delta y_{l,i} + b \Delta v_{m,j})$, where $a, b\in\{0,\pm 1\}$ and they are not both zero at the same time. This is valid in the general case of $\Lr>1, \Lt>1$, as well as in the special cases of $\Lr=1$ ($a=0$, $b=1$), or $\Lt=1$ ($a=1$, $b=0$). Since $w=(l,i)\in\PsetR$ and $\pPt=(m,j)\in\PsetT$, we have that $i\neq l$ and $j\neq m$, and $\Av[f(x, a \Delta y_{l,i} + b \Delta v_{m,j})]=0$ according to Lemma~\ref{lem:f:delta:zeromean} and~\eqref{def:Av:decomposition}. 
		We recall that $\Av[\cdot]$ is an integral over the same domain as over which the infimum is taken in~\eqref{eq:lem:inf:claim}, namely $\boldsymbol{y}\in [0,2\pi)^{\Lr},\boldsymbol{v}\in [0,2\pi)^{\Lt}$. Hence, $\Av[f(x, a \Delta y_{l,i} + b \Delta v_{m,j})]=0$ implies that
		\begin{align}
		0 &= \Av[\JA(\phiVrs,\boldsymbol{x},\boldsymbol{z},\boldsymbol{y},\boldsymbol{v})]\\
		&\ge \Av\left[\inf_{\boldsymbol{y}\in [0,2\pi)^{\Lr},\boldsymbol{v}\in [0,2\pi)^{\Lt}} \JA(\phiVrs,\boldsymbol{x},\boldsymbol{z},\boldsymbol{y},\boldsymbol{v})\right]\\
		&=(2\pi)^{\Lr + \Lt} \inf_{\boldsymbol{y}\in [0,2\pi)^{\Lr},\boldsymbol{v}\in [0,2\pi)^{\Lt}} \JA(\phiVrs,\boldsymbol{x},\boldsymbol{z},\boldsymbol{y},\boldsymbol{v}),
		\end{align}
		where the last equality holds since $\Av[C] = C (2\pi)^{\Lr + \Lt}$ for a constant $C$, and the lemma therefore follows.
	\end{proof}

	\begin{lemma}\label{lem:all0s}
		Let $\JA$ be as defined in~\eqref{eq:J:wp:compact}, where 
		$\boldsymbol{x}\in \mathbb{R}^{\Lr}$,  
		$\boldsymbol{y}\in [0,2\pi)^{\Lr}$,
		$\boldsymbol{z}\in \mathbb{R}^{\Lt}$,
		and $\boldsymbol{v}\in [0,2\pi)^{\Lt}$. Let $\mathcal{X^*}$ be as defined in~\eqref{eq:f:x:zeros:prop}, and 
			\begin{align}
		\DsetX&\triangleq\{\Delta x_{w},  \Delta z_{\pPt}, \Delta x_{w}\pm\Delta z_{\pPt}: w \in\PsetR, \pPt \in\PsetT\}\nonumber\\
		&= \{ \Delta z_{\pPt}: \pPt \in\PsetT, \PsetR=\emptyset \}\nonumber\\
		&= \{ \Delta x_{w}: w \in\PsetR, \PsetT=\emptyset \}.\label{def:Dset:X}
		\end{align}
		(i) If $\mathcal{D}_X\subset \mathcal{X^*}$ then
		\begin{equation}
		\JA(\phiVrs,\boldsymbol{x},\boldsymbol{z},\boldsymbol{y},\boldsymbol{v})= 0, \quad \boldsymbol{y}\in [0,2\pi)^{\Lr},\boldsymbol{v}\in [0,2\pi)^{\Lt}. \label{eq:lem:all0s:claim}
		\end{equation}
		(ii) Assuming that $c_w,c'_w,d_{\pPt},d'_{\pPt}>0$, $w\in \PsetR$, $\pPt\in\PsetT$, then~\eqref{eq:lem:all0s:claim} $\implies$ $\mathcal{D}_X\subset \mathcal{X^*}$.
		
	\end{lemma}
	
	\begin{proof}
		\textit{(i)} We see from~\eqref{eq:J:wp:compact} that $\JA$ is a linear combination of terms of the form $f(\Delta_X, \Delta_Y)$, where $\Delta_X\in\mathcal{D}_X$ defined in~\eqref{def:Dset:X}, and $\Delta_Y\in\mathcal{D}_Y$, 
\begin{align} 
			\mathcal{D}_Y &\triangleq \{\Delta y_w, \Delta v_{\pPt}, \Delta y_w \pm \Delta v_{\pPt}: w \in\PsetR, \pPt \in\PsetT\} \nonumber\\
			&= \{ \Delta v_{\pPt}: \pPt \in\PsetT, \PsetR=\emptyset \}\nonumber \\
			&= \{ \Delta y_{w}: w \in\PsetR, \PsetT=\emptyset \}.\label{def:Dset:Y}
			\end{align}

		
		 Suppose $\DsetX \subset\mathcal{X^{*}} $, then~$\Delta_X\in\DsetX \subset\mathcal{X}^*$, which together with~\eqref{eq:f:x:zeros:prop} implies that $f(\Delta_X, \Delta_Y) = 0$ for all $\Delta_Y\in\mathbb{R}$, which in turn implies that~\eqref{eq:lem:all0s:claim} hold. In other words, 
		$\DsetX \subset \mathcal{X^{*}} \Rightarrow\eqref{eq:lem:all0s:claim}$.
		
		\textit{(ii)} Suppose~\eqref{eq:lem:all0s:claim} holds, then $\JA^2 = 0$, which implies that $\Av[\JA^2]=0$.
		From~\eqref{eq:J:wp:compact}, we see that $\JA^2$ can be written as a linear combinations of 
		terms of the form $f(\Delta_X, \Delta_Y)f(\Delta_X', \Delta_Y')$, where $\Delta_X, \Delta_X'\in\DsetX$ and $\Delta_Y, \Delta_Y'\in\DsetY$. From Lemma~\ref{lem:f:delta:f:delta:zeromean}, we can see that  $\Av[f(\Delta_X, \Delta_Y)f(\Delta_X', \Delta_Y')]=0$ when $\Delta_Y\neq\Delta_Y'$. Hence, it follows from~\eqref{eq:J:wp:compact} that 
		\begin{align}
		\Av&[\JA^2]
		= \Av[\xi]+\nonumber\\
		\Av&\left[\sum_{w\in \PsetR} c_{w}^2 f^2(\Delta x_{w},\Delta y_{w}) 
		+ \sum_{\pPt\in \PsetT} d_{\pPt}^2 f^2(\Delta z_{\pPt},\Delta v_{\pPt})\right.\nonumber\\
		& + \sum_{w\in \PsetR}\sum_{\pPt\in \PsetT} \frac{(c'_{w}d'_{\pPt})^2}{4} f^2(\Delta         x_{w}+\Delta z_{\pPt},\Delta y_{w}+\Delta v_{\pPt}) \nonumber\\
		&\left. + \sum_{w\in \PsetR}\sum_{\pPt\in \PsetT} \frac{(c'_{w}d'_{\pPt})^2}{4} f^2(\Delta         x_{w}-\Delta z_{\pPt},\Delta y_{w}-\Delta v_{\pPt})\right],\label{eq:J2:wp:compact}
		\end{align}
		where $\Av[\xi] = 0$ accounts for all cross-terms, i.e., the terms when $\Delta_Y\neq\Delta_Y'$. 
		Note that in the special case of $\Lr=1$, $\PsetR=\emptyset$ and following the convention~\eqref{conv:ZeroSet:sum}, $A[\JA^2]=\sum_{\pPt\in \PsetT} d_{\pPt}^2~ \Av[f^2(\Delta z_{\pPt},\Delta v_{\pPt})]$, while for the special case $\Lt=1$, $A[\JA^2]=\sum_{w\in \PsetR} c_{w}^2~ \Av[f^2(\Delta x_{w},\Delta y_{w}) ]$.
		
		The squared $f$-functions in~\eqref{eq:J2:wp:compact} are of the form 
		$f^2(\Delta_X, \Delta_Y)$, where $\Delta_X\in\DsetX $ and $\Delta_Y\in\DsetY$. 
		We note that not all combinations of $\Delta_X\in\mathcal{D}_X$ and $\Delta_Y\in\DsetY$ are present in~\eqref{eq:J2:wp:compact}, but will not make this dependency notationally explicit. Now, with this abuse of notation, given~\eqref{eq:J2:wp:compact} and the assumption that $c_w$, $c'_w$, $d_{\pPt}$, $d'_{\pPt}>0$, are positive reals, $\Av[\JA^2]=0$ holds, if and only if, 
		\begin{equation}
		\label{eq:A:ff:zero}
		\Av[f^2(\Delta_X, \Delta_Y)]=0, \qquad\Delta_X\in\DsetX, \Delta_Y\in \DsetY.
		\end{equation}  
		Recalling~\eqref{eq:f:cases:app2}, $f^2(\Delta_X, \Delta_Y)$ is equal to
		\begin{align}\nonumber 
		\begin{cases}
		\displaystyle K^2 \cos^2(\Delta_Y),& \Delta_X\in\mathcal{X}\\
		\displaystyle \frac{\sin^2(K\Delta_X)}{\sin^2(\Delta_X)}\cos^2\big(\Delta_Y-(K-1)\Delta_X\big),& \Delta_X\notin\mathcal{X}.
		\end{cases}
		\end{align}
		We also recall that the $\Av[f^2(\Delta_X, \Delta_Y)]$ is an integral over the domain $\boldsymbol{y}\in[0, 2\pi)^{\Lr}, \boldsymbol{v}\in[0, 2\pi)^{\Lt}$, i.e., with respect to variables in $\Delta_Y$. Hence, $\Av[f^2(\Delta_X, \Delta_Y)]$ is proportional to the integral $\Av[\cos^2(\Delta_Y+ C)]$ for a constant $C$. Now, since $2\cos^2(a) = 1 + \cos(2a)$ and $\Av[1] = (2\pi)^{\Lt+\Lr}$,  
		\begin{align}\nonumber
		\Av[\cos^2(\Delta_Y+ C)]
		&= \frac{1}{2}(2\pi)^{\Lt+\Lr} + \frac{1}{2}\Av[\cos(2\Delta_Y+ 2C)].
		\end{align}
		The latter integral is zero, since $\Delta_Y \in \DsetY$
	contains at least two independent integration variables (e.g., $\Delta_Y=\Delta v_{\pPt}=\Delta v_{m,j}=v_m-v_j $), and therefore $\Av[\cos(2\Delta_Y+2C)]=0$. We have now shown that
		\begin{align}\nonumber
		\Av[f^2(\Delta_X, \Delta_Y)]=
		\begin{cases}
		\displaystyle \frac{1}{2}(2\pi)^{\Lt+\Lr}K^2 ,& \Delta_X\in\mathcal{X}\\
		\displaystyle \frac{1}{2}(2\pi)^{\Lt+\Lr}\frac{\sin^2(K\Delta_X)}{\sin^2(\Delta_X)},& \Delta_X\notin\mathcal{X}
		\end{cases}
		\end{align}
		which is zero if and only if, $\Delta_X\in\mathcal{X}^*$, 
		and thus~\eqref{eq:A:ff:zero} $\iff \mathcal{D}_X\subset\mathcal{X}^*$. 
		
		Putting our argument in order,~\eqref{eq:lem:all0s:claim}$~\Rightarrow\Av[\JA^2]=0$, and under the assumption that $c_w$, $c'_w$, $d_{\pPt}$, $d'_{\pPt}$ $>0$, $\Av[\JA^2]=0$ $\iff$~\eqref{eq:A:ff:zero}$\iff \mathcal{D}_X\subset\mathcal{X}^*$. Therefore, we  conclude that 
		$\eqref{eq:lem:all0s:claim}\Rightarrow\mathcal{D}_X\subset\mathcal{X}^*$, and the lemma follows.
	\end{proof}
		\begin{lemma}\label{lem:LrLt:K}
		Let $\boldsymbol{x}\in \mathbb{R}^{\Lr}$, $\boldsymbol{z} \in \mathbb{R}^{\Lt}$, 
		$\DsetX$ be as defined in~\eqref{def:Dset:X} and $\mathcal{X^{*}}$ as defined in~\eqref{eq:f:x:zeros:prop}. Then,
		\begin{align}\label{eq:lemLRLTK:iff}
		\exists (\boldsymbol{x},\boldsymbol{z}):\DsetX\subset \mathcal{X^{*}}\iff \Lr\Lt\leq K.
		\end{align}
		One selection of elements that satisfies $\DsetX\subset \mathcal{X^{*}}$ is 
		\begin{subequations}\label{eq:lem3:slopes}
			\begin{align}
			x_l&= l\frac{\pi}{K},\quad l=0,1,\cdots \Lr-1,\label{eq:lem3:slopes:x}\\
			z_m&=m \Lr \frac{\pi}{K}, \quad m=0,1, \cdots \Lt-1.\label{eq:lem3:slopes:z}
			\end{align}
		\end{subequations}
	\end{lemma}
	\begin{proof}
		Before tackling the proof, we start by presenting some definitions and notation. We recall that
		$	\mathcal{X^*}=\{q\pi/K: q\in \mathbb{Z}\}\setminus\{n\pi,~n\in \mathbb{Z}\}$.
		Then, let us define the set
		\begin{align}\label{def:X:Dstar}
		\mathcal{{X}^{**}}\triangleq\{a\pi/K:a\in \mathbb{Z}, 1\leq a\leq K-1\}, 
		\end{align}
		which satisfies $\mathcal{{X}^{**}}\subset \mathcal{X^*} $ and
		\begin{align}\label{eq:aXstar:modXdS}
		a \in  \mathcal{X^*} &\iff \rem( a,\pi) \in \mathcal{X^{**}},\\
		(a- b)\in \mathcal{X^*}  &\implies \rem(a,\pi) \neq \rem(b,\pi).\label{eq:lem:proof:prop:aPbInX}
		\end{align} 
		where $\rem(a,b)$ is the remainder of dividing $a$ by $b$. 
		The sum or subtraction of two sets is defined as 
		\begin{align}
		\mathcal{B}\pm \mathcal{C}=\{b\pm c:b\in\mathcal{B}, c\in \mathcal{C}\}.
		\end{align}	
		Then, we define for $\boldsymbol{x}$ and $\boldsymbol{z}$ the sets
			\begin{align}
			\mathcal{E}_{\Delta x}&\triangleq\{\Delta x_{l,i}: (l,i)\in \PsetR\}, & \text{If }\PsetR=\emptyset,~\mathcal{E}_{\Delta x}=\emptyset.\\
			\mathcal{E}_{\Delta z}&\triangleq\{\Delta z_{m,j}: (m,j)\in \PsetT \}, &\text{If }\PsetT=\emptyset,~\mathcal{E}_{\Delta z}=\emptyset.
			\end{align}
		From~\eqref{def:Dset:X} we see that, when $\Lr>1$, $\Lt>1$ ($\PsetR\neq\emptyset$, $\PsetT\neq\emptyset$)
		\begin{align}\label{eq:DsetX:union}
		\DsetX=\mathcal{E}_{\Delta x}\cup \mathcal{E}_{\Delta z}\cup (\mathcal{E}_{\Delta x}+\mathcal{E}_{\Delta z})\cup (\mathcal{E}_{\Delta x}-\mathcal{E}_{\Delta z} ),
		\end{align}
		while $\DsetX=\mathcal{E}_{\Delta z}$, when $\Lr=1$, and $\DsetX=\mathcal{E}_{\Delta x}$ when $\Lt=1$. Last, we introduce the set notation
		\begin{align}
		\{N\}_{n_0}^{n_1} \triangleq \{N\in \mathbb{Z}: n_0\leq N\leq n_1\}.
		\end{align}
		
		Now, we are set to tackle the proof. We start by showing that $\Lr\Lt\leq K \implies \exists(\boldsymbol{x},\boldsymbol{z}):\DsetX\subset \mathcal{X^{*}}$.
		Let $\Lr\Lt\leq K$, and let $\boldsymbol{x}$ and $\boldsymbol{z}$ be selected following~\eqref{eq:lem3:slopes}. 
	Then, for $(l,i)\in \PsetR$ and $(m,j)\in \PsetT$, we have $\Delta x_{l,i}=x_i-x_l=(i-l)\pi/K$, $0\leq l<i\leq \Lr-1$, and 
 $\Delta z_{m,j}=z_j-z_m=(j-m)\Lr\pi/K$, $0\leq m<j\leq \Lt-1$. This implies 
		\begin{align}\label{eq:EdeltaX}
		\mathcal{E}_{\Delta x}&= \{a\pi/K: 1\leq a\leq \Lr-1\},\\
		\mathcal{E}_{\Delta z}& =\{b\Lr\pi/K: 1\leq b\leq \Lt-1\}, \label{eq:EdeltaZ}\\
		\mathcal{E}_{\Delta z}+\mathcal{E}_{\Delta x}&=\label{eq:EdeltaXpZ}
		\big\{(b\Lr+a)\frac{\pi}{K}:  \{a\}_1^{\Lr-1},  \{b\}_1^{\Lt-1}    \big\},\\
		\mathcal{E}_{\Delta z}-\mathcal{E}_{\Delta x}&=\label{eq:EdeltaXmZ}
		\big\{  (b\Lr-a)\frac{\pi}{K}:\{a\}_1^{\Lr-1},  \{b\}_1^{\Lt-1}   \big\}.
		\end{align}
		Since,  $(\Lr-1)\leq K-1 $, from~\eqref{eq:EdeltaX} and~\eqref{def:X:Dstar} we deduce that $\mathcal{E}_{\Delta x} \subset 	\mathcal{{X}^{**}} $. 
		Similarly, we have that in~\eqref{eq:EdeltaZ}--\eqref{eq:EdeltaXmZ}
		\begin{align}
		\Lr&\leq b\Lr~\quad \leq (\Lt-1)\Lr\qquad \leq K-1,\\
		\Lr+1&\leq b\Lr+a\leq \Lr\Lt-1 \qquad\leq K-1,\\
		1&\leq	b\Lr-a\leq (\Lt-1)\Lr -1\leq K-1.
		\end{align}
		 It consequently follows by~\eqref{def:X:Dstar} that, also, $\mathcal{E}_{\Delta z}$, $\mathcal{E}_{\Delta z}+ \mathcal{E}_{\Delta x}$, and $\mathcal{E}_{\Delta z}- \mathcal{E}_{\Delta x}$ are subsets of $\mathcal{X}^{**}\subset  \mathcal{X^*}$. Note that $\mathcal{E}_{\Delta z}- \mathcal{E}_{\Delta x}\subset\mathcal{X^*} $ implies that $-(\mathcal{E}_{\Delta z}- \mathcal{E}_{\Delta x})\subset \mathcal{X^*}$ (if $x\in \mathcal{X^*}$, $-x\in \mathcal{X^*}$). 
		Recalling~\eqref{eq:DsetX:union} we can conclude that $\DsetX \subset\mathcal{X}^{*} $.
		Note that this hold in the special case of $\Lr=1$ ($\mathcal{E}_{\Delta x}=\emptyset$), because $\DsetX=\mathcal{E}_{\Delta z}\subset\mathcal{X}^{*}$, also it holds when $\Lt=1$ ($\mathcal{E}_{\Delta z}=\emptyset$), for $\DsetX=\mathcal{E}_{\Delta x}\subset\mathcal{X}^{*}$.
		
		We have shown that when $\Lr\Lt\leq K$,~\eqref{eq:lem3:slopes} is a solution satisfying $\DsetX \subset\mathcal{X}^{*} $, when $\Lr+\Lt> 2$, and thus the statement "$\Lr\Lt\leq K \implies \exists(\boldsymbol{x},\boldsymbol{z}):\DsetX\in \mathcal{X^{*}}$" holds.
		
		Now we move to the second part of the claim of the lemma, $\exists(\boldsymbol{x},\boldsymbol{z}):\DsetX\in \mathcal{X^{*}}\implies \Lr\Lt\leq K$.
		We shall show that through a proof of contradiction. 
	    Assume 
		that $\Lr\Lt> K$. By construction of the set $\mathcal{X^{**}}$~\eqref{def:X:Dstar} we have that
		\begin{align}\label{eq:Lem:optCon:Proof:contradiction:K}
		|\mathcal{X^{**}}|=K-1<\Lr\Lt-1,\qquad K<\Lr\Lt.
		\end{align} 
		To form a contradiction, we proceed to find a lower bound on the cardinality of the set $\mathcal{X^{**}}$ assuming that there exists a solution $(\boldsymbol{x},\boldsymbol{z})$ that satisfies $\DsetX\subset \mathcal{X^{*}}$. From~\eqref{eq:DsetX:union} we see that this last condition implies that $\mathcal{E}_{\Delta x}$, $\mathcal{E}_{\Delta z}$, and $\mathcal{E}_{\Delta x}\pm \mathcal{E}_{\Delta z}$ are subsets of  $\mathcal{X^{*}}$.  
		Without loss of generality\footnote{If $x_0\neq0$, we define $x_0',x_1',\cdots,x_{\Lr-1}'$ such that $x_0'=0$ and $x_l'=x_l-x_0$. If $z_0\neq0$, we analogously define $z_0',z_1',\cdots,z_{\Lt-1}'$. Since, $\Delta x_{l,i}=\Delta x'_{l,i}$, and $\Delta z_{m,j}=\Delta z'_{m,j}$ then $(\{x'_l\},\{z_m'\})$ are also solutions, i.e., $\mathcal{E}_{\Delta x'}\cup \mathcal{E}_{\Delta z'}\cup (\mathcal{E}_{\Delta x'}\pm \mathcal{E}_{\Delta z'}) \subset \mathcal{{X}^{*}}$.
		} let $x_0=0$, and $z_0=0$. Then, we have $\mathcal{E}_{\Delta x} \subset \mathcal{{X}^{*}}$ is equivalent to 
		\begin{align}
		\Delta x_{0,i}=(x_i-x_0)&=x_i \in  \mathcal{X^*},~ 1\leq i\leq \Lr-1\label{eq:lem:optCon:Proof:x:sub1}\\
		\Delta x_{l,i}=	(x_i-x_l) &\in  \mathcal{X^*} ,~1\leq l<i\leq \Lr-1 \label{eq:lem:optCon:Proof:x:sub2}
		\end{align} 
		Then, from~\eqref{eq:aXstar:modXdS} and~\eqref{eq:lem:optCon:Proof:x:sub1} we get that
		\begin{align}\nonumber
		\mathcal{\tilde{E}}_{x}=\big\{\rem(x_i,\pi): \{i\}_1^{\Lr-1}\big\} \subset  \mathcal{X^{**}}.
		\end{align}
		Moreover, since~\eqref{eq:lem:proof:prop:aPbInX} and~\eqref{eq:lem:optCon:Proof:x:sub2}, implies that $\rem(x_i,\pi) \neq\rem(x_l,\pi)$, $ i\neq l$, then  $|\mathcal{\tilde{E}}_{x}|=\Lr-1$.
		
		Similarly, employing $\mathcal{E}_{\Delta z} \subset \mathcal{{X}^{*}}$ we deduce that 
		\begin{align}\nonumber
		\mathcal{\tilde{E}}_{z}=\big\{\rem (z_j,\pi):  \{j\}_1^{\Lt-1}\big\} \subset  \mathcal{X^{**}},~|\mathcal{\tilde{E}}_{z}|=\Lt-1.
		\end{align}
		By the previous analysis, if only $ \mathcal{E}_{\Delta x} \subset \mathcal{{X}^{*}}$ and $\mathcal{E}_{\Delta z} \subset \mathcal{{X}^{*}}$ 
		need to be satisfied then $|\mathcal{X^{**}}|\geq \max\{\Lt-1,\Lr-1\}$.
		However, let us consider $(\mathcal{E}_{\Delta x}\pm \mathcal{E}_{\Delta z})\subset \mathcal{{X}^{*}}$, which can be expressed as
		\begin{align}\label{eq:DxpDz:General}
		(\Delta x_{l,i}\pm \Delta z_{m,j}) \in \mathcal{X^*},~ \begin{cases} 
		0\leq l< i\leq \Lr-1 \\
		0\leq m< j\leq \Lt-1
		\end{cases}
		\end{align}
		In particular, using $\Delta x_{0,i} - \Delta z_{0,j}= (x_{i}-z_j) \in \mathcal{X^*}$, 
		which by~\eqref{eq:lem:proof:prop:aPbInX} implies that
		\begin{align}
		\rem(x_{i},\pi)&\neq \rem(z_j,\pi),\quad  \{i\}_1^{\Lr-1}, \{j\}_1^{\Lt-1},
		\end{align}
		we get 
		$\mathcal{\tilde{E}}_x \cap \mathcal{\tilde{E}}_z=\emptyset$ and $|\mathcal{X^{**}}|\geq (\Lt-1 +\Lr-1)$.
		This lower bound can be further improved considering a set that contains the elements $\rem(x_i+z_j,\pi)$.
		From~\eqref{eq:DxpDz:General}, we get that for $\{i\}_1^{ \Lr-1}$, $\{j\}_1^{\Lt-1}$ 
		\begin{align}\label{eq:x_p_z_in_X}
		\Delta x_{0,i}+\Delta z_{0,j}=(x_i+z_j) \in  \mathcal{X^{*}}.
		\end{align}

		Moreover, for $\{i\}_1^{\Lr-1}$, $1\le m<j\leq \Lt-1$, 
		\begin{align}
		\Delta x_{0,i}+\Delta z_{m,j}=(x_i+z_j-z_m)&\in \mathcal{X^*}  \overset{\eqref{eq:lem:proof:prop:aPbInX}}{\implies}\nonumber\\ \rem(x_i+z_j,\pi)&\neq \rem(z_m,\pi),\label{eq:con:elaborate:1}\\
		\Delta x_{0,i}-\Delta z_{m,j}=(x_i-z_j+z_m)&\in \mathcal{X^{*}} \overset{\eqref{eq:lem:proof:prop:aPbInX}}{\implies}\nonumber\\\rem(x_i+z_m,\pi)&\neq \rem(z_j,\pi).\label{eq:con:elaborate:2}
		\end{align}
		Combining these last two, and since by~\eqref{eq:lem:optCon:Proof:x:sub1} it holds that $\rem(x_i+z_j,\pi)\neq \rem(z_j,\pi)$, $\{i\}_1^{\Lr-1}, \{j\}_1^{\Lt-1}$, we deduce that $\forall i,j,n\geq 1$
			\begin{align}\label{eq:x+z_neq_z}
		\rem(x_i+z_j,\pi)\neq \rem(z_n,\pi).
		\end{align}
		We can use similar elaboration based on the conditions $(\Delta x_{l,i}\pm\Delta z_{0,j})\in \mathcal{X^*}$, $1\le l<i\leq \Lr-1$, and $\{j\}_1^{\Lt-1}$ to deduce that $\forall i,j,k\geq 1 $
		\begin{align}\label{eq:x+z_neq_x}
		\rem(x_i+z_j,\pi)\neq \rem(x_{k},\pi).
		\end{align}
		Finally, elaborating the conditions $(\Delta x_{l,i}\pm\Delta z_{m,j})\in \mathcal{X^*}$, where $1\le l<i\leq \Lr-1$,  $1\le m<j\leq \Lt-1$, and employing the fact that $\rem(x_{i}+ z_{j},\pi) \neq \rem(x_i+z_n,\pi) $, $j\neq n$, and $\rem(x_{i}+ z_{j},\pi) \neq \rem(x_k+z_j,\pi)  $, $i \neq k$, we can readily reach that $\forall i,j,k,n\geq 1$, $(i,j)\neq (k,n)$
		\begin{align}\label{eq:lem:optCon:proof:subxz}	
		\rem(x_{i}&+ z_{j},\pi) \neq \rem(x_k+z_n,\pi). 
		\end{align}
		
		Now, let us gather our findings. Using~\eqref{eq:aXstar:modXdS} and~\eqref{eq:x_p_z_in_X} we deduce that 
		\begin{align}
		\mathcal{\tilde{E}}_{(x+z)}=\big\{\rem(x_i+z_j,\pi): \{ i\}_1^{\Lr-1}, \{j\}_1^{\Lt-1} \big  \} \subset \mathcal{X^{**}},\nonumber
		\end{align}
		while~\eqref{eq:x+z_neq_z} and~\eqref{eq:x+z_neq_x}
		implies that 
		$\mathcal{\tilde{E}}_{(x+z)}\cap (\mathcal{\tilde{E}}_x\cup \mathcal{\tilde{E}}_z)=\emptyset$,
		and, finally, by~\eqref{eq:lem:optCon:proof:subxz} we get that $|\mathcal{\tilde{E}}_{(x+z)}|=(\Lt-1)(\Lr-1)$.

		Then, under the assumption that there exists a solution $(\boldsymbol{x},\boldsymbol{z})$ satisfying $\DsetX\subset  \mathcal{{X}^{*}}$, we have shown that 	$\mathcal{\tilde{E}}_{x}$, $\mathcal{\tilde{E}}_{z}$ and $\mathcal{\tilde{E}}_{(x+z)}$ are disjoint subsets of $\mathcal{X^{**}}$, and thus 
		\begin{align}
		|\mathcal{X^{**}}|\geq \Lr-1 +\Lt-1+(\Lr-1)(\Lt-1)=\Lr\Lt-1.\nonumber
		\end{align}
		This contradicts the claim in~\eqref{eq:Lem:optCon:Proof:contradiction:K}. Note that in the special case of $\Lr=1$, $\DsetX=\mathcal{E}_{\Delta z}\subset \mathcal{X^*}$, which implies $\mathcal{\tilde{E}}_{z}\subset \mathcal{X^{**}}$, $|\mathcal{\tilde{E}}_{z}|=\Lt-1$, and thus $|\mathcal{X^{**}}|\geq \Lt-1$. In similar manner, we can deduce that $|\mathcal{X^{**}}|\geq \Lr-1$ when $\Lt=1$. Both form a contradiction to~\eqref{eq:Lem:optCon:Proof:contradiction:K}. 
		Therefore, we readily conclude that for $\Lr+\Lt>2$ the statement "$\exists (\boldsymbol{x},\boldsymbol{z}): \DsetX\in \mathcal{{X}^{*}}$ $\implies \Lr\Lt\leq K$" holds, and by this we come to end the proof of the lemma.
		\color{black}
	\end{proof}

	Now we show the claims of the theorem. 
	\begin{proof}
Given the objective function $\SsnrA$~\eqref{eq:SsnrA:xy}, for
		$\boldsymbol{x}\in \mathbb{R}^{\Lr}$,  
		$\boldsymbol{z}\in \mathbb{R}^{\Lt}$, 
		$\boldsymbol{y}\in [0,2\pi)^{\Lr}$,
		and $\boldsymbol{v}\in [0,2\pi)^{\Lt}$,	we express
	\begin{align}
	\SsnrA^\star(\phiVrs)&=\sup_{\boldsymbol{x},\boldsymbol{z}}\inf_{\boldsymbol{y},\boldsymbol{v}}\SsnrA(\phiVrs,\boldsymbol{x},\boldsymbol{z},\boldsymbol{y},\boldsymbol{v} )\nonumber \\
	&= K\overline{G}(\phiVrs)+ \sup_{\boldsymbol{x},\boldsymbol{z}}\inf_{\boldsymbol{y},\boldsymbol{v}}\JA(\phiVrs,\boldsymbol{x},\boldsymbol{z},\boldsymbol{y},\boldsymbol{v} ).\label{eq:Sstar:eq:Jstar}
	\end{align}
	\begin{enumerate}
		\item From Lemma~\ref{lem:inf}~\eqref{eq:lem:inf:claim} we get that
		\begin{align}\label{eq:proof:theo:last:general}
		\sup_{\boldsymbol{x},\boldsymbol{z}}~\inf_{\boldsymbol{y},\boldsymbol{v}}\JA(\phiVrs,\boldsymbol{x},\boldsymbol{z},\boldsymbol{y},\boldsymbol{v} )\leq 0,
		\end{align}
		thus $\SsnrA^\star(\phiVrs)\leq K\overline{G}(\phiVrs)$, and~\eqref{eq:theo:1:UB} follows \big($\phiVrs=(\phir,\phit)$\big).
\item 		If $\Lr\Lt\leq K$, Lemma~\ref{lem:LrLt:K}~\eqref{eq:lemLRLTK:iff} indicate that $\exists (\boldsymbol{\tilde{x}},\boldsymbol{\tilde{z}}):\DsetX \subset\mathcal{X}^*$ which by Lemma~\ref{lem:all0s}~\textit{(i)} implies that 
		\begin{align}\label{eq:forReamrk}
		\JA(\phiVrs,\boldsymbol{\tilde{x}},\boldsymbol{\tilde{z}},\boldsymbol{y},\boldsymbol{v} )=0, \quad \forall \boldsymbol{y},\boldsymbol{v}.
		\end{align}
		Hence, the upper bound in~\eqref{eq:proof:theo:last:general} is achievable, that is
		\begin{align}\label{eq:proof:theo:last}
		\sup_{\boldsymbol{x},\boldsymbol{z}}~\inf_{\boldsymbol{y},\boldsymbol{v}}\JA(\phiVrs,\boldsymbol{x},\boldsymbol{z},\boldsymbol{y},\boldsymbol{v} )= 0.
		\end{align}
		This, together with~\eqref{eq:Sstar:eq:Jstar}, implies that $\SsnrA^\star(\phiVrs )=K\overline{G}(\phiVrs)$, hence ~\eqref{eq:theo:1} holds. 
		\item Since~\eqref{eq:proof:theo:last} is guaranteed when $(\boldsymbol{\tilde{x}},\boldsymbol{\tilde{z}})$ satisfies $\DsetX \subset\mathcal{X}^*$ then this last is is a sufficient optimality condition when $\Lr\Lt\leq K$. Recalling that $x_l=\phr_lT/2$, $z_m=\pht_mT/2$, and the definition of $\DsetX$~\eqref{def:Dset:X}, we can straightforwardly see that $\DsetX \subset\mathcal{X}^*$ is the same as~\eqref{eq:th:optCon}, when $\Lr>1, \Lt>1$, while in the special cases of $\Lr=1$ or $\Lt=1$, it is expressed as~\eqref{eq:th:optCon:z} or~\eqref{eq:th:optCon:x}, respectively.
		
		\item Assume that~\eqref{eq:proof:theo:last} holds \big(i.e., $\SsnrA^\star(\phiVrs )=K\overline{G}(\phiVrs)$\big) with solution $(\boldsymbol{\bar{x}},\boldsymbol{\bar{z}})$. Then, for $\boldsymbol{y}\in[0,2\pi)^{\Lr}, \boldsymbol{v}\in[0,2\pi)^{\Lt}$
	\begin{align}\label{eq:nonzeroJA}
	   \JA(\phiVrs,\boldsymbol{\bar{x}},\boldsymbol{\bar{z}},\boldsymbol{y},\boldsymbol{v})&\geq \inf_{\boldsymbol{y},\boldsymbol{v}} \JA(\phiVrs,\boldsymbol{\bar{x}},\boldsymbol{\bar{z}},\boldsymbol{y},\boldsymbol{v})=0.
	\end{align}
		Also, by Lemma~\ref{lem:inf}~\eqref{eq:lem:inf:AJeq:zero} we get that
		\begin{align}\label{eq:JA:xbar:zero}
		    \Av[\JA(\phiVrs,\boldsymbol{\bar{x}},\boldsymbol{\bar{z}},\boldsymbol{y},\boldsymbol{v})]=0.
		\end{align}
		
		From~\eqref{eq:nonzeroJA} and~\eqref{eq:JA:xbar:zero}, we see that over the intervals $\boldsymbol{y}\in[0,2\pi)^{\Lr}$ and $\boldsymbol{v}\in[0,2\pi)^{\Lt}$, the continuous function $\JA(\phiVrs,\boldsymbol{\bar{x}},\boldsymbol{\bar{z}},\boldsymbol{y},\boldsymbol{v})$ is a non-negative function that integrates to zero, and thus it must satisfy~\eqref{eq:lem:all0s:claim}. It therefore, follows that~\eqref{eq:proof:theo:last} $\implies$~\eqref{eq:lem:all0s:claim}.
		
		 Now, given that $|\gr_l(\phir)|,|\gt_m(\phit)|>0$, $\forall l,m$, then from~\eqref{def:J:cli:dmj} we see that $c_{w},c'_{w},d_{\pPt},d'_{\pPt} >0$ for $w=(l,i) \in\PsetR$, $\pPt=(m,j)\in \PsetT$. Under these assumptions, Lemma~\ref{lem:all0s}~\textit{(ii)} tells us that~\eqref{eq:lem:all0s:claim} implies that $(\boldsymbol{\bar{x}},\boldsymbol{\bar{z}})$ satisfies $\DsetX \subset\mathcal{X}^*$, which is again the same as~\eqref{eq:th:optCon}. 
		 The condition $\Lr\Lt\leq K$, follows as an implication by Lemma~\ref{lem:LrLt:K}.
		
		
	\item (Corollary~\ref{cor:ABN:ph_slopes}) By Lemma~\ref{lem:LrLt:K}, we know that~\eqref{eq:lem3:slopes}, which is the same as~\eqref{eq:th:ph_slopes}, satisfies $\DsetX \subset\mathcal{X}^*$ and thus it is optimal. Interchanging between $\Lr$ and $\Lt$ in~\eqref{eq:lem3:slopes}, and interpreting ($x'_l=\pht_lT/2$, $z'_m=\phr_mT/2$) instead of ($x_l=\phr_lT/2$, $z_m=\pht_mT/2$) we get that~\eqref{eq:th:ph_slopes:reciprocal} is also an optimal selection of phase slopes.
	\end{enumerate}
	We have shown that all claims of Theorem~\ref{theo:main} hold, including Remark~\ref{remark:theorem:main} which follows from~\eqref{eq:forReamrk} and~\eqref{eq:SsnrA:xy}. Also, we have shown that the claim of Corollary~\ref{cor:ABN:ph_slopes} holds, and thus the proof is completed.
\end{proof}

\section{Proof of Theorem~\ref{th:SA}}\label{app:theo:SA}
\begin{proof}
We start by defining,  $x_l \triangleq \Lt \phr_lT/2$, $y_l \triangleq \psr_l$, where $0\le l\le\Lr-1 $, $e_m \triangleq |\gt_m(\phit)|^2$, $0\leq m\leq\Lt-1$, and $\phiVrs \triangleq(\phir,\phit)$.
Then, from~\eqref{def:S:SA} we write $\JB$ as
\begin{align}
&\JB(\phiVrs,\boldsymbol{x},\boldsymbol{y})\nonumber\\
&=\sum_{m=0}^{\Lt-1} e_m\sum_{w\in \PsetR}c'_{w}\sum_{k=0}^{\Kr-1}\cos\big(\Delta y_{w}-2\Delta x_{w}(m+k\Lt)/\Lt\big)\nonumber\\
&=\sum_{w\in \PsetR}c'_{w} \sum_{m=0}^{\Lt-1} e_m f_{\Kr}(\Delta x_{w},\Delta y_{w}-2\Delta x_{w}m/\Lt ),\label{eq:JB:compact}
\end{align} 
where $f_{\Kr}$ is defined according to~\eqref{eq:f:def} over a sum of $\Kr=K/\Lt \in\mathbb{Z}$ terms (instead of $K$), $\PsetR$ is defined in~\eqref{eq:Pr:def}, $c'_w$ is defined in~\eqref{def:J:cli:dmj}, 
$\boldsymbol{x}=[x_0,x_1,\cdots,x_{\Lr-1}]^{\textsf{T}}\in \mathbb{R}^{\Lr}$, and $\boldsymbol{y}=[y_0,y_1,\cdots,y_{\Lr-1}]^{\textsf{T}}\in [0,2\pi)^{\Lr}$. 
All properties and lemmas stated for $f$ in Appendix~\ref{app:f:properties} are valid for $f_{\Kr}$. In particular we can define a set $\Xstb$ according to~\eqref{eq:f:x:zeros:prop}, 
\begin{align}
\Xstb=\{q\pi/\Kr:q\in \mathbb{Z} \}\setminus \mathcal{X},\label{def:Xbarstar:app}
\end{align}
where $\mathcal{X}=\{q\pi:q\in \mathbb{Z}\}$.


We observe that similarly to $\JA$ (Theorem~\ref{theo:main}), $\JB$ is also a linear combination of $f_{\Kr}$-functions with non-negative coefficients. The arguments stated earlier to show Theorem~\ref{theo:main} can be used here to show the result of Theorem~\ref{th:SA}. Therefore, to save on space and avoid repetition we just outline the proof steps. To show the claims of Theorem~\ref{th:SA} we can take the following steps.
\begin{enumerate}
	\item Using a similar argument to what was used in Lemma~\ref{lem:inf} we can easily show that 
	\begin{align}\label{eq:Ay:Jb:eqZero}
	\Av_{\boldsymbol{y}}[\JB(\phiVrs,\boldsymbol{x},\boldsymbol{y})]&=0,\\
		\inf_{\boldsymbol{y}\in [0,2\pi)^{\Lr}} \JB(\phiVrs,\boldsymbol{x},\boldsymbol{y})&\leq 0. \label{eq:inf:JB:leq0}
	\end{align}
	where $\Av_{\boldsymbol{y}}[h(\cdot)]=\int_{[0, 2\pi)^{\Lr}} h(\cdot) d\boldsymbol{y}$. 
	\item By using similar steps as followed in Lemma~\ref{lem:all0s}, we can augment our finding by showing that 
	
	a) Given $\mathcal{D}'_{X}=\{\Delta x_{w}:w\in \PsetR \}$, 
	\begin{align}\label{eq:XinXstar:imp:JBallzeros}
	 \mathcal{D}'_{X}\subset \Xstb \Longrightarrow \JB(\phiVrs,\boldsymbol{x},\boldsymbol{y})=0,~\boldsymbol{y}\in [0,2\pi)^{\Lr}.
	\end{align}

	b) Assuming that $c'_w,e_m>0$, $\forall w, m$ and $e_m\neq C$, $\forall m$,
	\begin{align}\label{eq:XinXstar:impR:JBallzeros}
	\mathcal{D}'_{X} \subset \Xstb  \Longleftarrow \JB(\phiVrs,\boldsymbol{x},\boldsymbol{y})=0,~
	\boldsymbol{y}\in [0,2\pi)^{\Lr}.
	\end{align}
	
	\item Last, by repeating some of the arguments used to show Lemma~\ref{lem:LrLt:K} we can show that
	\begin{align}\label{eq:X:iff:LrK}
	\exists \boldsymbol{x}:\mathcal{D}'_{X} \subset \Xstb \iff \Lr\leq \Kr,
	\end{align}
	and a solution that satisfies $\mathcal{D}'_{X} \subset \Xstb$ is given by
	\begin{align}\label{eq:phses:x:B}
	x_l=l \pi/\Kr,~ l=0, 1, \cdots, \Lr-1, ~\Lr\leq \Kr.
	\end{align}
    (An alternative approach is to directly use Lemma~\ref{lem:LrLt:K} for the special case of $\hat{L}_{\textrm{s}}=1$, which implies $\mathcal{D}_{X}=\mathcal{D}'_{X}$, to prove~\eqref{eq:X:iff:LrK} and ~\eqref{eq:phses:x:B}.)
\end{enumerate}
Then, from~\eqref{def:S:SA} and~\eqref{def:Sbstar} we have that
\begin{align}
\SsnrB^\star(\phiVrs)=K\overline{G}(\phiVrs) +
\sup_{\boldsymbol{x}\in\mathbb{R}^{\Lr}}\inf_{\boldsymbol{y}\in[0,2\pi)^{\Lr}} \JB(\phiVrs,\boldsymbol{x},\boldsymbol{y}).\label{eq:theorem2:proff:Sxy}
\end{align}
The claims of Theorem~\eqref{th:SA} are shown as follows.
\begin{enumerate}
     \item Using~\eqref{eq:inf:JB:leq0} we deduce that $\SsnrB^\star(\phiVrs)\leq K\overline{G}(\phiVrs)$.
   \item Employing~\eqref{eq:XinXstar:imp:JBallzeros} and~\eqref{eq:X:iff:LrK} we can deduce that $\SsnrB^\star(\phiVrs)=K\overline{G}(\phiVrs)$ when $\Lr\leq \Kr$, and $\mathcal{D}'_{X}=\{\Delta x_{w}:w\in \PsetR \}\subset \Xstb$ is a sufficient optimality condition. 
    
     \item Given that $\SsnrB^\star(\phiVrs)=K\overline{G}(\phiVrs)$ with a solution $\bar{\boldsymbol{x}}$. Then, $\JB(\phiVrs,\bar{\boldsymbol{x}},\boldsymbol{y})\geq\inf_{\boldsymbol{y}} \JB(\phiVrs,\bar{\boldsymbol{x}},\boldsymbol{y})= 0$. Adding this to~\eqref{eq:Ay:Jb:eqZero}, and by continuity of $\JB$, we deduce that $\JB(\phiVrs,\bar{\boldsymbol{x}},\boldsymbol{y})=0$, $\boldsymbol{y}\in [0,2\pi)^{\Lr} $, i.e., the right-hand side of~\eqref{eq:XinXstar:impR:JBallzeros} holds. 
    Now, assuming that $|\gr_l(\phir)|, |\gt_m(\phit)|>0$, $\forall l, m$, and $|\gt_m(\phit)|\neq C$, $\forall m$, it follows that $c'_w, e_m>0$, $\forall w, m$ and $e_m\neq C$, $\forall m$. Then, by~\eqref{eq:XinXstar:impR:JBallzeros}
    we deduce that if $\SsnrB^\star(\phiVrs)=K\overline{G}(\phiVrs)$ then $\mathcal{D}'_{X} \subset \Xstb$. The condition $\Lr\leq \Kr$ follows as a consequence of~\eqref{eq:X:iff:LrK}.
    \item Last, since the solution in~\eqref{eq:phses:x:B} satisfies $\mathcal{D}'_{X} \subset \Xstb$, then it is optimal, and by that the proof outline of the theorem is completed.
     
\end{enumerate}\end{proof}
	\section{Proof of Lemma~\ref{lem:Multiport:bound}}\label{app:multiport}
	\begin{proof}
We start by defining $\Lt$-vectors, $\xtV\in \mathbb{R}^{\Lt}$ and $\vtV\in[0,2\pi)^{\Lt}$ with elements $\xt_m\triangleq \pht_{m}T/2$ and $\vt_l\triangleq \pst_{m}$ where $0 \leq m \leq \Lt-1$.
Also, we define for every port $p=0, 1, \cdots, P-1$, the vectors $\xrVp\in \mathbb{R}^{\Lrp}$ and $\yrVp \in [0,2\pi)^{\Lrp}$ with elements $\xr_{l}^{(p)}\triangleq \phr_{l,p}T/2$, $\yr_{l}^{(p)}\triangleq \psr_{l,p}$, where $0\leq l\leq \Lrp-1$. Recalling that $\sum_{p}\Lrp=\Lr$, we define  $\Lr$-vector $\xrV=[\xrV^{(0)}, \xrV^{(1)}, \cdots, \xrV^{(P-1)}]^{\textsf{T}} \in \mathbb{R}^{\Lr}$ and, similarly, we define $\yrV=[\yrV^{(0)}, \yrV^{(1)}, \cdots, \yrV^{(P-1)}]^{\textsf{T}} \in[0,2\pi)^{\Lr}$. 
Then, we can express $\SsnrD$ and $\SsnrAp$ defined in~\eqref{def:SsnrD} and~\eqref{def:SsnrAp}, respectively, as
\begin{align}
&\SsnrD(\phiVrs,\xrV,\xtV,\yrV,\vtV )= \sum_{p=0}^{P-1} \SsnrAp (\phiVrs,\xrVp,\xtV,\yrVp,\vtV )\nonumber\\
&=K\sum_{p=0}^{P-1}\overline{G}_p(\phiVrs)+
\sum_{p=0}^{P-1} \JAp(\phiVrs,\xrVp,\xtV,\yrVp,\vtV ),\label{eq:Sd:eqGplusJ}
\end{align}
where $\phiVrs=(\phir,\phit)$, and $\overline{G}_p$, $\JAp$ are expressed based on $\overline{G}$~\eqref{def:G} and $\JA$~\eqref{def:J:original}, with $\Lr$ substituted by $\Lrp$ and $\gr_{l}$ by $\gr_{l,p}$. In particular, $\JAp$ can be expressed based on $\JA$~\eqref{eq:J:wp:compact}, where $\PsetR$ is defined with respect to $\Lrp$ instead of $\Lr$ ($\xrVp$ and $\yrVp$ have $\Lrp$ elements), and $c_{w},c'_{w},d_{\pPt},d'_{\pPt}$ are attributed an additional sub-index $p$. 
Note that $\JA$ and $\JAp$ represent the same function, and the difference 
between the two is in notation.
Using Lemma~\ref{lem:inf}~\eqref{eq:lem:inf:AJeq:zero}, and recalling~\eqref{eq:A:def} we get that 
\begin{align}
\int_{\mathcal{I}}\JAp(\phiVrs,\xrVp,\xtV,\yrVp,\vtV)d\yrVp d\vtV=0,~\forall p
\end{align}
where $\mathcal{I}=[0,2\pi)^{\Lrp+\Lt}$. Then, by the linearity of integral we can easily deduce that 
\begin{align}\label{eq:Sd:int:zero}
\int_{\mathcal{J}} \sum_{p=0}^{P-1}\JAp(\phiVrs,\xrVp,\xtV,\yrVp,\vtV) d\yrV d\vtV=0,
\end{align}
where $\mathcal{J}=[0,2\pi)^{\Lr+\Lt}$, $\sum_{p}\Lrp=\Lr$.
Finally, by monotonicity of Riemann integral and~\eqref{eq:Sd:int:zero} we can conclude that for  $\xrV\in\mathbb{R}^{\Lr}$, $\xtV \in \mathbb{R}^{\Lt}$
\begin{align}
\inf_{\yrV\in [0,2\pi)^{\Lr},\vtV\in [0,2\pi)^{\Lt}} \bigg[\sum_{p=0}^{P-1} \JAp (\phiVrs,\xrVp,\xtV,\yrVp,\vtV)\bigg]\leq 0,\nonumber
\end{align}
which together with~\eqref{eq:Sd:eqGplusJ} allow us to deduce that~\eqref{eq:UB:Mp} holds, and this ends the proof of the lemma.
\end{proof}
\end{appendices}
\bibliographystyle{ieeetr}
\bibliography{paperII_ref}

\end{document}